%% file: 2Dcase_arxiv04.tex
\documentclass[10pt]{article}
\usepackage{amssymb,amsbsy}
\usepackage{amsmath}
\usepackage{amsthm}
\usepackage{graphics,graphicx}
\usepackage[T1]{fontenc}
\usepackage{color}

\input{definitions}

\title{Symmetric and skew-symmetric weight functions in perturbation models of 2D interfacial cracks}

\author{A. Piccolroaz$^{(1)}$, G. Mishuris$^{(2)}$, A.B. Movchan$^{(3)}$
\\
\\
$^{(1)}$
{\it Dipartimento di Ingegneria Meccanica e
Strutturale, Universit\`a di Trento,}
\\ {\it Via Mesiano 77, I-38050 Trento, Italia}
\\
{\it
$^{(2)}$
Wales Institute of Mathematical and Computational Sciences, }
\\ {\it Institute of Mathematical and Physical Sciences, Aberystwyth University, }
\\ {\it Ceredigion SY23 3BZ, Wales U.K.,}
\\
{\it
$^{(3)}$
Department of Mathematical Sciences, University of Liverpool, }
\\ {\it Liverpool L69 3BX, U.K.}
}

\begin{document}

\maketitle

\begin{abstract}
\noindent
In this paper we address the vector problem of a 2D half-plane interfacial crack loaded by a general asymmetric distribution of forces acting on its faces. It is 
shown that the general integral formula for the evaluation of stress intensity factors, as well as high-order terms, requires both symmetric and skew-symmetric weight 
function matrices. The symmetric weight function matrix is obtained via the solution of a Wiener-Hopf functional equation, whereas the derivation of the skew-symmetric 
weight function matrix requires the construction of the corresponding full field singular solution.

The weight function matrices are then used in the perturbation analysis of a crack advancing quasi-statically along the interface between two dissimilar media. A general 
and rigorous asymptotic procedure is developed to compute the perturbations of stress intensity factors as well as high-order terms.
\end{abstract}

\newpage

\tableofcontents

\newpage

\section{Introduction}
The theory of weight functions is fundamental in the evaluation of stress intensity factors for asymptotic representations near non-regular boundaries such as cusps, 
wedges, cracks. The classical work of Bueckner defined weight functions for several types of cracks, both in 2D and 3D, as the stress intensity factors corresponding to 
the point force loads applied to the faces of the cracks placed in a homogeneous continuum (Bueckner, 1987, 1989). In particular, for a penny shaped crack in three 
dimensions, as well as a half-plane crack, Bueckner has introduced the weight functions corresponding to a general distribution of forces on the opposite faces of 
the crack, which also included the case of forces acting in the same direction. 

In the present paper, we use the term ``symmetric'' load for forces acting on opposite crack faces and in opposite directions, whereas the term ``skew-symmetric'' or 
``anti-symmetric'' load will be used for forces acting on opposite crack faces but in the same direction, see Fig.~\ref{fig01a}. Consequently, symmetric loads are always 
self-balanced, whereas skew-symmetric loads are not, in general (however, we shall consider only the case of self-balanced loading). Note that this terminology is 
different from that usually in use, where the words ``symmetric'' and ``anti-symmetric'' are intended in terms of the symmetry about the plane containing the crack, 
so that, for example, shear forces acting on opposite crack faces but in the same direction are considered to be symmetric (see for example Meade and Keer, 1984).
%%%%%%%%%%%%%%%%%%%%%%%%%%%%%%%%%%%%%%%%%%%%%%%%%%%%%%%%%%%%%%%%%%%%%%
\begin{figure}[!htb]
\begin{center}
\vspace*{3mm}
\includegraphics[width=10cm]{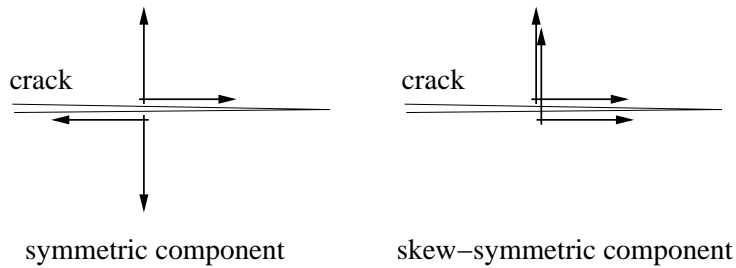}
\caption{\footnotesize Symmetric and skew-symmetric parts of the loading in the case of single point forces. In practical cases, the loading is 
given as a combination of point forces and/or distributed forces in such a way it is self-balanced (see for example the problem analysed in 
Section \ref{example}).}
\label{fig01a}
\end{center}
\end{figure}
%%%%%%%%%%%%%%%%%%%%%%%%%%%%%%%%%%%%%%%%%%%%%%%%%%%%%%%%%%%%%%%%%%%%%%

Although in two dimensions skew-symmetric loading does not contribute to the stress intensity factors, it becomes essential in three dimensions. Both symmetric and 
skew-symmetric loads on the faces of a half-plane crack in a three dimensional homogeneous elastic space are also considered by Meade and Keer (1984).

The situation when the crack is placed at an interface between two dissimilar elastic media is substantially different from the analogue corresponding to a homogeneous 
elastic space with a crack. Here, even in the case of two dimensions (plane strain or plane stress), the stress components oscillate near the crack tip and also the 
skew-symmetric loads generate non-zero stress intensity factors. Also for the Mode III, where the stress components do not oscillate, there is a non-vanishing  
skew-symmetric component of the weight function. 

Symmetric twodimensional weight functions (i.e. stress intensity factors for symmetric opening loads applied on the crack faces) for interfacial cracks were analysed by 
Hutchinson, Mear and Rice (1987). The problem of symmetric three dimensional Bueckner's weight functions for interfacial cracks has been addressed by Lazarus and Leblond 
(1998) and Piccolroaz {\em et al.}\ (2007). On the basis of these results, Pindra {\em et al.}\ (2008) studied the evolution in time of the deformation of the front of a 
semi-infinite 3D interface crack propagating quasistatically in an infinite heterogeneous elastic body.

However, to our best knowledge, the skew-symmetric weight function for interfacial cracks has never been constructed for the plane strain case nor even for the Mode III 
deformation. On the other hand, in most applications, especially for the interfacial crack, the loading is not symmetric and consequently the effects of the skew-symmetric 
loading in the analysis of fracture propagation require a thoroughly investigation.

The aim of this paper is to construct the aforementioned skew-symmetric weight functions for the twodimensional interfacial crack problem. To this purpose, we develop a 
general approach, which allows us to obtain the general weight functions, as used by Willis and Movchan (1995), defined as non-trivial singular solutions of a boundary 
value problem for interfacial cracks with zero tractions on the crack faces but unbounded elastic energy. By taking the trace of these general weight functions on the 
real axis, one can arrive to the notion of Bueckner's weight function associated with the point force load on the crack faces (Piccolroaz {\em et al.}, 2007). It is then 
shown that the skew-symmetric part of the loading contributes to the stress singularity at the crack edge and thus to the resulting stress intensity factors. These 
results are presented in Section \ref{wfunc} for the plane strain problem and in Section \ref{antiplane} for the Mode III problem.

We summarize that the symmetric weight function matrix $\jump{0.15}{\bU}$ for a plane strain interfacial crack has the form (compare with \eq{jumpu} in the main text):
\beq
\lb{symm}
\barr{ll}
\ds \jump{0.15}{\bU}(x_1) = \frac{1}{2 d_0 \sqrt{2 \pi x_1}} \left\{ \frac{x_1^{-i \epsilon}}{c_1^+} \bmB + \frac{x_1^{i \epsilon}}{c_1^-} \bmB^\top \right\} 
& \text{for } x_1 > 0, \\[5mm]
\ds \jump{0.15}{\bU}(x_1) = 0 & \text{for } x_1 < 0,
\earr
\eeq
where $i$ is the imaginary unit, $\bmB = \begin{pmatrix} 1 & -i \\ i & 1 \end{pmatrix}$, the superscript $^\top$ denote transposition, and the bimaterial parameters 
$\epsilon$ and $d_0$ are defined in Appendix \ref{app1}. The quantities $c_1^+$, $c_1^-$ appear as coefficients in the representation of oscillatory terms in the 
asymptotics of stress near the tip of the interfacial crack, as defined in the sequel (see \eq{stress1}--\eq{stress2}). 

The result for the skew-symmetric weight function matrix $\langle {\bf U} \rangle$, unpublished in the earlier literature, reads as follows (compare with \eq{meanu} 
in the main text):
\beq
\lb{skew}
\barr{ll}
\ds \langle \bU \rangle(x_1) = \frac{\Ga}{2} \jump{0.15}{\bU}(x_1) & \text{for } x_1 >0, \\[5mm]
\ds \langle \bU \rangle(x_1) = -i \frac{\alpha (d_* - \gamma_*)}{4d_0^3 \sqrt{-2 \pi x_1}} 
\left\{ \frac{(-x_1)^{-i \epsilon}}{c_1^+} \bmB - \frac{(-x_1)^{i \epsilon}}{c_1^-} \bmB^\top \right\} & \text{for } x_1 < 0,
\earr
\eeq
where the bimaterial constant $\gamma_*$ and the Dundurs parameters $\alpha$ and $d_*$ are defined in Appendix \ref{app1}. The conventional notations 
$\jump{0.15}{f}(x) = f(x,0^+) - f(x,0^-)$ and $\langle f \rangle(x) = \frac{1}{2}(f(x,0^+) + f(x,0^-))$ are in use here to denote the symmetric and skew-symmetric 
components, respectively.

Note that the result for the symmetrical weight functions is consistent with the representation for the stress intensity factors derived by Hutchinson, Mear and Rice 
(1987).

For the Mode III case, the scalar symmetric and skew-symmetric weight functions are (compare with \eq{jumpu3} and \eq{meanu3} in the main text):
\beq
\jump{0.15}{U_3}(x_1) = \left\{
\barr{ll}
\ds \frac{1-i}{\sqrt{2\pi}} x_1^{-1/2}, & \text{for } x_1 > 0, \\[3mm]
0, & \text{for } x_1 < 0,
\earr \right., \quad
\langle U_3\rangle = \frac{\eta}{2} \jump{0.15}{U_3},
\eeq
where $\eta$ is another bimaterial constant defined in Appendix \ref{app1}.

In Section \ref{identity} we use both symmetric and skew-symmetric weight function matrices together with the Betti identity in order to derive the stress intensity 
factors for an interfacial crack subjected to a general load. We summarize here that the integral formula for the computation of the complex stress intensity factor 
in the plane strain problem, $K = K_\modI + i K_\modII$, has the form (compare with \eq{SIF} in the main text):
\beq
\bK = -i \bmM_1^{-1} \lim_{x_1'\to 0} \int_{-\infty}^{0} \left\{\jump{0.15}{\bU}^\top(x_1'-x_1) \bR \langle \bp \rangle(x_1) 
+ \langle\bU\rangle^\top(x_1'-x_1) \bR \jump{0.15}{\bp}(x_1) \right\} dx_1,
\eeq
where $\bK = [K,K^*]^\top$, the superscript $^*$ denotes conjugation, $\bmM_1$ is defined in \eq{emme1} and \eq{stress2}, $\bR$ is defined in \eq{erre} and 
$\langle \bp \rangle(x_1)$, $\jump{0.15}{\bp}(x_1)$ stand for the symmetric and skew-symmetric parts of the loading, respectively.

For the computation of the Mode III stress intensity factor, $K_\modIII$, the integral formula has the form (compare with \eq{k3} in the main text):
\beq
K_\modIII = - \sqrt{\frac{2}{\pi}} \int_{-\infty}^{0} \left\{ \langle p_3 \rangle(x_1) + \frac{\eta}{2} \jump{0.15}{p_3}(x_1) \right\} (-x_1)^{-1/2} dx_1.
\eeq

Needless to say, by replacing the loading with the Dirac delta function, these integral formulae immediately give the Bueckner's weight functions. Also, this approach 
allows us to derive the constants in the high-order terms, which are essential in the perturbation analysis (Lazarus and Leblond, 1998; Piccolroaz {\em et al.}, 2007). 

Section \ref{perturbation} is devoted to the perturbation model of a crack advancing quasi-statically along the interface between two dissimilar media. We develop and prove 
an alternative approach consistent with the idea originally presented in Willis and Movchan (1995). Both symmetric and skew-symmetric weight function matrices are 
essential in this perturbation analysis.

A simple illustrative example is given in Section \ref{example}, which shows quantitatively the influence of the skew-symmetric loading for 
interfacial cracks.

Section \ref{prel_analysis} is introductory and contains auxiliary results concerning the physical fields around the interfacial crack tip, which are necessary for the analysis presented 
in Sections \ref{wfunc}--\ref{antiplane}.

Finally, the Appendix \ref{app1} collects all bimaterial parameters involved in the analysis of the interfacial crack, including the well-known parameters as well as 
the new parameters appearing in the skew-symmetric problem. The parameters have been classified according to whether they are involved in the symmetrical or 
skew-symmetrical weight functions. The Appendix \ref{app3} provides the theorems required for proving the approach presented in Section \ref{perturbation}.

\section{Preliminary results for an interfacial crack}
\lb{prel_analysis}
To introduce the main notations and the mathematical framework for our model, we start with the analysis of the displacement and stress fields around a semi-infinite 
plane crack located on an interface between two dissimilar isotropic elastic media, with elastic constants denoted by $\nu_\pm,\mu_\pm$, see Fig. \ref{fig02a}.
%%%%%%%%%%%%%%%%%%%%%%%%%%%%%%%%%%%%%%%%%%%%%%%%%%%%%%%%%%%%%%%%%%%%%%
\begin{figure}[!htb]
\begin{center}
\vspace*{3mm}
\includegraphics[width=9cm]{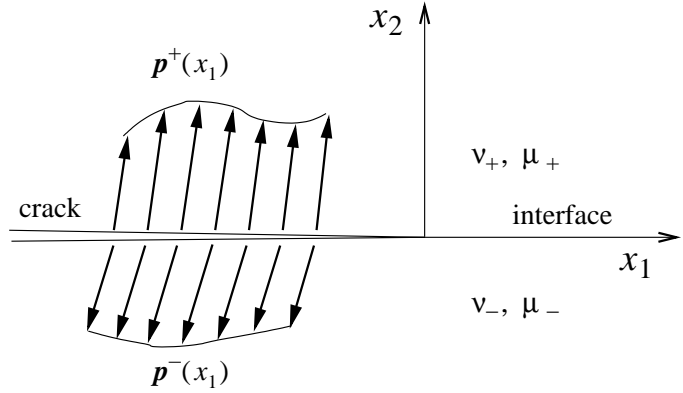}
\caption{\footnotesize Geometry of the model.}
\label{fig02a}
\end{center}
\end{figure}
%%%%%%%%%%%%%%%%%%%%%%%%%%%%%%%%%%%%%%%%%%%%%%%%%%%%%%%%%%%%%%%%%%%%%%

The loading is given by tractions acting upon the crack faces. In terms of a Cartesian coordinate system with the origin at the crack tip, traction components are 
defined as follows 
\beq
\lb{load}
\sigma_{2j}^\pm(x_1,0^\pm) = p_j^\pm(x_1), \quad \text{for } x_1 < 0, \quad j = 1,2,
\eeq
where $p_j^\pm(x_1)$ are prescribed functions.

The load is assumed to be self-balanced, so that its principal force and moment vectors are equal to zero. We also assume that the forces are applied outside a 
neighbourhood of the crack tip and vanish at infinity (specific assumptions on the behaviour of the loading will be discussed in the sequel). The body forces are 
assumed to be zero. It is convenient to use the notion of the symmetric and skew-symmetric parts of the loading
\beq
\lb{parts}
\langle p_j \rangle(x_1) = \frac{p_j^+(x_1) + p_j^-(x_1)}{2}, \quad \jump{0.15}{p_j}(x_1) = p_j^+(x_1) - p_j^-(x_1), \quad j = 1,2.
\eeq

The solution to the physical problem is sought in the class of functions which vanish at infinity and have finite elastic energy.

The results presented in Sections \ref{mellin} and \ref{asymptotics} are known in the literature (see for example: Williams, 1959; Erdogan, 1963; Rice, 1965; 
Willis, 1971), at least for symmetric loading. These results, together with the results on higher order terms presented in Section \ref{high}, which are less common 
in the literature, provide the formulae needed for the perturbation analysis developed in the sequel of the paper.

\subsection{Representation in terms of Mellin transforms}
\lb{mellin}
In polar coordinates the Mellin transforms for the 
displacement vector and the stress tensor with locally bounded elastic energy are defined as follows 
$$
\tilde{\bu}(s,\theta) = \int_0^\infty \bu(r,\theta) r^{s-1} dr, \quad
\tilde{\bsigma}(s,\theta) = \int_0^\infty \bsigma(r,\theta) r^s dr, 
$$
and they are represented by analytic functions in the strips $-\vartheta_0 < \Re(s) < \vartheta_\infty$ and $-\gamma_0 < \Re(s) < \gamma_\infty$, respectively, where 
$\vartheta_0,\vartheta_\infty \ge 0$ ($\vartheta_0 + \vartheta_\infty > 0$), $\gamma_0,\gamma_\infty > 0$ are constants related to the behaviour of the solution at the 
crack tip and at infinity, namely
\beq
\lb{exponents}
\bu(r,\theta) = 
\left\{
\barr{ll}
O(r^{\vartheta_0}), & r \to 0, \\
O(r^{-\vartheta_\infty}), & r \to \infty, 
\earr
\right.
\quad
\bsigma(r,\theta) = 
\left\{
\barr{ll}
O(r^{\gamma_0-1}), & r \to 0, \\
O(r^{-\gamma_\infty-1}), & r \to \infty, 
\earr
\right.
\eeq

Correspondingly, the inverse transforms are 
\beq
\lb{inverse}
\bu(r,\theta) = 
\frac{1}{2\pi i} \int_{\omega_1 - i\infty}^{\omega_1 +i\infty} \tilde{\bu}(s,\theta) 
r^{-s} ds, \quad
\bsigma(r,\theta) = 
\frac{1}{2\pi i} \int_{\omega_2 - i\infty}^{\omega_2 + i\infty} \tilde{\bsigma}(s,\theta) 
r^{-s - 1} ds, 
\eeq
where $-\vartheta_0 < \omega_1 < \vartheta_\infty$ and $-\gamma_0 < \omega_2 < \gamma_\infty$.

\subsubsection{General solution to the field equations of linear elasticity} 
\lb{general}
In polar coordinates, with the origin at the crack tip, the Mellin transforms of stresses and displacements satisfy the relations (the material constants involved in 
these representations are reported in Appendix \ref{app1})
$$
\tilde{\sigma}_{\theta\theta}^\pm = 
C_1^\pm \cos[(s + 1)\theta] + C_2^\pm \cos[(s - 1)\theta] + C_3^\pm \sin[(s + 1)\theta] 
+ C_4^\pm \sin[(s - 1)\theta],
$$
$$
\tilde{\sigma}_{rr}^\pm = -\frac{1}{s}\left[\tilde{\sigma}_{\theta\theta}^\pm 
- \frac{1}{s-1}\nderiv{}{\theta}{2}\tilde{\sigma}_{\theta\theta}^\pm\right], \quad
\tilde{\sigma}_{r\theta}^\pm = \frac{1}{s-1} \deriv{}{\theta}\tilde{\sigma}_{\theta\theta}^\pm,
$$
$$
\tilde{u}_r^\pm = \frac{1}{2s \mu_\pm}[\tilde{\sigma}_{\theta\theta}^\pm 
- (1 - \nu_\pm)\tilde{\sigma}_0^\pm], \quad
\tilde{u}_\theta^\pm = -\frac{1}{2s \mu_\pm}\left[\tilde{\sigma}_{r\theta}^\pm 
+ \frac{1 - \nu_\pm}{s + 1}\deriv{}{\theta}\tilde{\sigma}_0^\pm\right], \quad
$$
where $\tilde{\sigma}_0^\pm = \tilde{\sigma}_{rr}^\pm + \tilde{\sigma}_{\theta\theta}^\pm$.

Hence, we deduce
$$
\tilde{\sigma}_{rr}^\pm = 
-\frac{s + 3}{s - 1} C_1^\pm \cos[(s + 1)\theta] - C_2^\pm \cos[(s - 1)\theta]
- \frac{s + 3}{s - 1} C_3^\pm \sin[(s + 1)\theta] - C_4^\pm \sin[(s - 1)\theta],
$$
$$
\tilde{\sigma}_{r\theta}^\pm = 
-\frac{s + 1}{s - 1} C_1^\pm \sin[(s + 1)\theta] - C_2^\pm \sin[(s - 1)\theta]
+ \frac{s + 1}{s - 1} C_3^\pm \cos[(s + 1)\theta] + C_4^\pm \cos[(s - 1)\theta],
$$
$$\barr{l}
\ds
\tilde{u}_r^\pm = 
\frac{1}{2s \mu_\pm} \left\{ C_1^\pm \cos[(s + 1)\theta] + C_2^\pm \cos[(s - 1)\theta]
+ C_3^\pm \sin[(s + 1)\theta] + C_4^\pm \sin[(s - 1)\theta] \right. \\[3mm]
\ds \hspace{60mm}
\left. + \frac{4(1 - \nu_\pm)}{s - 1}
\left[ C_1^\pm \cos[(s + 1)\theta] + C_3^\pm \sin[(s + 1)\theta] \right] \right\},
\earr
$$
$$
\barr{l}
\ds
\tilde{u}_\theta^\pm = 
-\frac{1}{2s \mu_\pm} \left\{ -C_1^\pm \frac{s + 1}{s - 1} \sin[(s + 1)\theta] 
- C_2^\pm \sin[(s - 1)\theta]
+ C_3^\pm \frac{s + 1}{s - 1} \cos[(s + 1)\theta] \right. \\[3mm]
\ds \hspace{30mm}
\left. + C_4^\pm \cos[(s - 1)\theta] + \frac{4(1 - \nu_\pm)}{s - 1}
\left[ C_1^\pm \sin[(s + 1)\theta] - C_3^\pm \cos[(s + 1)\theta] \right] \right\}.
\earr
$$

The coefficients $C_j^\pm$ depend on the loading and will be given in the next section.

\subsubsection{The boundary conditions and full field solution} 
\lb{full}
The coefficients $C_j^\pm$ are obtained from the boundary conditions on the crack faces, namely 
$$
\sigma_{\theta\theta}^\pm(r,\pm\pi) = p^\pm(r), \quad \sigma_{r\theta}^\pm(r,\pm\pi) = q^\pm(r),
$$
where $p^\pm(r), q^\pm(r)$ are prescribed functions given by $p^\pm(r) = p_2^\pm(-r)$, $q^\pm(r) = p_1^\pm(-r)$, see \eq{load}. Mishuris and Kuhn (2001) give the result in a 
compact form as the solution to the following system of algebraic equations
\beq
\lb{system}
\frac{2\sin \pi s}{s - 1}
\left[
\barr{l}
C_1^\pm(s) \\[3mm]
C_3^\pm(s) 
\earr
\right] = 
\left[
\barr{ll}
-\sin \pi s & \pm \cos \pi s \\[3mm]
\pm \cos \pi s & \sin \pi s
\earr
\right]
\bD(s)
\pm \left[
\barr{l}
\ds \langle \tilde{q} \rangle \pm \frac{1}{2} \jump{0.15}{\tilde{q}} \\[3mm]
\ds \langle \tilde{p} \rangle \pm \frac{1}{2} \jump{0.15}{\tilde{p}} 
\earr
\right],
\eeq
$$
\left[
\barr{l}
\ds
C_1^\pm(s) + C_2^\pm(s) \\[3mm]
\ds
\frac{s + 1}{s - 1} C_3^\pm(s) + C_4^\pm(s)
\earr
\right] = \bD(s),
$$
where $\bD(s) = - \bPhi^{-1}(s) \bF(s)$ and 
$$
\bPhi(s) = 
\left[
\barr{ll}
\cos \pi s & d_* \sin \pi s \\[3mm]
-d_* \sin \pi s & \cos \pi s 
\earr
\right],
\quad
\bF(s) = 
\left[
\barr{l}
\ds \langle \tilde{p} \rangle(s) + \frac{\alpha}{2} \jump{0.15}{\tilde{p}}(s) \\[3mm]
\ds \langle \tilde{q} \rangle(s) + \frac{\alpha}{2} \jump{0.15}{\tilde{q}}(s)
\earr
\right].
$$

Note that the prescribed boundary conditions appear in \eq{system} and in the definition of the vector $\bF(s)$, in terms of the symmetric and skew-symmetric parts of the 
loading. The material constants $d_*$ and $\alpha$ are the Dundurs parameters, as described in the Appendix \ref{app1}.

Denoting the determinant of $\bPhi(s)$ by $\delta(s) = \cos^2 \pi s + d_*^2 \sin^2 \pi s$, we obtain the full field solution in terms of the Mellin transform as follows 
$$\barr{l}
\ds 
C_1^\pm(s) = \frac{s - 1}{2 \delta(s)}
\Bigg\{
\langle \tilde{p} \rangle[(1 \mp d_*) \cos \pi s] 
+ \frac{1}{2} \jump{0.15}{\tilde{p}}[(1 \mp d_*)\alpha \cos \pi s] \\[3mm]
\ds \hspace{15mm} 
+ \langle \tilde{q} \rangle\left[-(1 \mp d_*)d_* \sin \pi s \right]
+ \frac{1}{2} \jump{0.15}{\tilde{q}}\left[(d_* - \alpha)d_* \sin \pi s 
+ (1 \mp \alpha)\frac{\cos^2 \pi s}{\sin \pi s}\right]
\Bigg\},
\earr
$$
$$
\barr{l}
\ds 
C_3^\pm(s) = \frac{s - 1}{2 \delta(s)}
\Bigg\{
\langle \tilde{p} \rangle\left[-(1 \mp d_*)d_* \sin \pi s\right]
+ \frac{1}{2} \jump{0.15}{\tilde{p}}\left[(d_*-\alpha)d_* \sin \pi s 
+ (1 \mp \alpha)\frac{\cos^2 \pi s}{\sin \pi s}\right] \\[3mm]
\ds \hspace{15mm}
+ \langle \tilde{q} \rangle[-(1 \mp d_*) \cos \pi s] 
+ \frac{1}{2} \jump{0.15}{\tilde{q}}[-(1 \mp d_*)\alpha \cos \pi s] 
\Bigg\},
\earr
$$
$$
\barr{l}
\ds 
C_2^\pm(s) = -\frac{1}{2 \delta(s)}
\Bigg\{
\langle \tilde{p} \rangle[(1 + s \pm (1 - s)d_*) \cos \pi s] 
+ \frac{1}{2} \jump{0.15}{\tilde{p}}[(1 + s \pm (1 - s)d_*)\alpha \cos \pi s] \\[3mm]
\ds \hspace{40mm} 
+ \langle \tilde{q} \rangle\left[(\pm(1 + s) + (1 - s)d_*)d_* \sin \pi s\right] \\[3mm]
\ds \hspace{40mm}
+ \frac{1}{2} \jump{0.15}{\tilde{q}}\left[-(\alpha(1 + s) + d_*(1 - s))d_* \sin \pi s 
- (1 \mp \alpha)(1 - s)\frac{\cos^2 \pi s}{\sin \pi s}\right]
\Bigg\},
\earr
$$
$$
\barr{l}
\ds 
C_4^\pm(s) = -\frac{1}{2 \delta(s)}
\Bigg\{
\langle \tilde{p} \rangle\left[(1 - s \pm (1 + s)d_*)d_* \sin \pi s\right] \\[3mm]
\ds \hspace{15mm}
+ \frac{1}{2} \jump{0.15}{\tilde{p}}\left[(\alpha(1 - s) + d_*(1 + s))d_* \sin \pi s 
+ (1 \mp \alpha)(1 + s)\frac{\cos^2 \pi s}{\sin \pi s}\right] \\[3mm]
\ds \hspace{15mm}
+ \langle \tilde{q} \rangle[(1 - s \pm (1 + s)d_*) \cos \pi s] 
+ \frac{1}{2} \jump{0.15}{\tilde{q}}[(1 - s \pm (1 + s)d_*)\alpha \cos \pi s] 
\Bigg\}.
\earr
$$

We note that some of the poles of the functions $C_j^\pm(s)$ are obtained from the solution of the equation $\delta(s) = \cos^2 \pi s + d_*^2 \sin^2 \pi s = 0$, so that 
the poles are given by 
\beq
\lb{poles}
s_n^\pm = \frac{1 - 2n}{2} \pm i\epsilon,
\eeq
where $n$ is an integer, and $\epsilon$ is the bimaterial constant defined in Appendix \ref{app1}.

For the purpose of evaluation of residues, it will be useful to have at hand the derivative of $\delta(s)$. This is given by 
$\delta'(s) = 2\pi(-1+d_*^2) \sin \pi s \cos \pi s$, and its value at $s = s_n^\pm$ is as follows
\beq
\lb{deltap}
\delta'\left(s_n^\pm\right) = \pm 2\pi i d_*.
\eeq
This formula will be used in the next section for the derivation of the asymptotic estimates of stress and displacement fields near the crack tip.

\subsection{Asymptotic representations near the crack tip and the stress intensity factors}
\lb{asymptotics}
The solution outlined above leads to the asymptotic representations of stress and displacement near the crack tip. Analysis of the singular terms in the stress components 
yields the complex stress intensity factor for the interfacial crack, where the stress singularity is accompanied by the oscillatory behaviour of the physical fields.

\subsubsection{Asymptotics of stress and displacement near the crack tip}
Taking into account the assumptions on the applied forces, the inspection of the results for $\tilde{\bsigma}(s,\theta)$ shows that $\tilde{\bsigma}(s,\theta)$ is 
analytic in the strip $-1/2 < \Re(s) < 1/2$. Therefore, in \eq{exponents} we have $\gamma_0 = \gamma_\infty = 1/2$, and choosing $\omega_2 = 0$, the inverse transform 
is given by 
$$
\bsigma(r,\theta) = \frac{1}{2\pi i} \int_{0-i\infty}^{0+i\infty} 
\tilde{\bsigma}(s,\theta) r^{-s-1} ds.
$$

By means of the Cauchy's residue theorem, we can get the asymptotics of $\bsigma(r,\theta)$ as $r \to 0$ as follows
$$
\bsigma(r,\theta) =  \sum_{\pm} \Res[\tilde{\bsigma}(s,\theta) r^{-s-1},s=s_1^\pm] + 
\frac{1}{2\pi i} \int_{\omega-i\infty}^{\omega+i\infty} \tilde{\bsigma}(s,\theta) r^{-s-1} ds, 
$$
where $\omega < -1/2$. The first two terms are the leading term, of the order $O(r^{-1/2})$, and the last term is a higher order term, of the order 
$O(r^{\beta}), \beta > -1/2$. The notation $\Res[f(s),s=s^*]$ stands for the residue of $f$ at the pole $s = s^*$,
$$
\Res[f(s),s=s^*] = \frac{1}{(m-1)!} \left. \frac{d^{m-1} f(s)(s-s^*)^m}{ds^{m-1}}\right|_{s=s^*},
$$
where $m$ is the order of the pole.

It suffices now to compute the residues at the poles $s = s_1^\pm = -1/2 \pm i\epsilon$,
$$
\barr{l}
\Res\left[\tilde{\bsigma}(s,\theta)\, r^{-s - 1}, s = -1/2 \pm i\epsilon\right] = \\[3mm]
\hspace{30mm}
\ds
= \left[\tilde{\bsigma}(s,\theta) \delta(s)\, r^{-s - 1}\right]_{s = -1/2 \pm i\epsilon} 
\cdot \Res\left[\frac{1}{\delta(s)}, s = -1/2 \pm i\epsilon\right] \\[3mm]
\hspace{30mm}
\ds
= \left[\tilde{\bsigma}(s,\theta) \delta(s)\, r^{-s - 1}\right]_{s = -1/2 \pm i\epsilon} 
\cdot \lim_{s \to -1/2 \pm i\epsilon} \frac{s-(-1/2 \pm i\epsilon)}{\delta(s)} \\[3mm]
\hspace{30mm}
\ds
= \left[\tilde{\bsigma}(s,\theta) \delta(s)\, r^{-s - 1}\right]_{s = -1/2 \pm i\epsilon} 
\cdot \frac{1}{\delta'(-1/2 \pm i\epsilon)} \\[3mm]
\hspace{30mm}
\ds
= \mp \frac{i}{2\pi d_*} 
\left[\tilde{\bsigma}(s,\theta) \delta(s)\right]_{s = -1/2 \pm i\epsilon} \cdot 
r^{-1/2 \mp i\epsilon},
\earr
$$
where in the last equality we used the formula \eq{deltap}.

It can be shown that $\beta = 0$, so that the leading term asymptotics of the stress field is 
\beq
\lb{sigma}
\bsigma(r,\theta) =  \beth_{-1/2}(r,\theta) + O(1),
\eeq
where
$$
\beth_{-1/2}(r,\theta) = 
\frac{i}{2\pi d_*}
\left\{-\left[\tilde{\bsigma}(s,\theta) \delta(s)\right]_{s = -1/2 + i\epsilon} 
r^{-1/2 - i\epsilon}
+ \left[\tilde{\bsigma}(s,\theta) \delta(s)\right]_{s = -1/2 - i\epsilon} 
r^{-1/2 + i\epsilon}\right\}.
$$

An inspection of the results for $\tilde{\bu}(s,\theta)$ shows that it is analytic in the same strip $-1/2 < \Re(s) < 1/2$, except for a simple pole in $s = 0$. Since 
we seek the solution vanishing at infinity, the strip of analyticity for the function $\tilde{\bu}(s,\theta)$ is $0 < \Re(s) < 1/2$. Therefore, in \eq{exponents} we 
have $\vartheta_0 = 0$, $\vartheta_\infty = \gamma_\infty = 1/2$. Choosing $\omega_1$ in this interval and applying the Cauchy's residue theorem, we can write the 
leading term asymptotics of $\bu(r,\theta)$ as $r \to 0$ as follows 
$$
\bu(r,\theta) = \Res\left[\tilde{\bu}(s,\theta)\, r^{-s}, s = 0\right] + \sum_{\pm} \Res\left[\tilde{\bu}(s,\theta)\, r^{-s}, s = s_1^\pm\right] 
+ \frac{1}{2 \pi i} \int_{\omega-i\infty}^{\omega+i\infty} \tilde{\bu}(s,\theta) r^{-s} ds,
$$
where $\omega < -1/2$. The first term of the order $O(1)$ corresponds to the translation of the crack tip and reads
\beq
\lb{qtheta}
\bu(0,\theta) = \bV_0(\theta) =  
\bQ(\theta)\bw_0, \quad \bQ(\theta) = \left[\barr{ll} \cos\theta & \sin\theta \\[3mm] -\sin\theta & \cos\theta \earr\right],
\eeq
where
$$
\barr{ll}
\ds w_{01} & \ds = u_1(0,0) \\[3mm]
& \ds = \frac{1}{2\pi\mu_\pm} \left\{ [1 - 2\nu_\pm \pm 2d_*(\nu_\pm - 1)] \pi \int_{0}^{\infty} \langle p \rangle(r) dr + 
(-1 \pm \alpha)(\nu_\pm - 1) \int_{0}^{\infty} \jump{0.15}{q}(r) (\log r) dr \right\},
\earr
$$
$$
\barr{ll}
\ds w_{02} & \ds = u_2(0,0) \\[3mm]
& \ds = \frac{1}{2\pi\mu_\pm} \left\{ -[1 - 2\nu_\pm \pm 2d_*(\nu_\pm - 1)] \pi \int_{0}^{\infty} \langle q \rangle(r) dr + 
(-1 \pm \alpha)(\nu_\pm - 1) \int_{0}^{\infty} \jump{0.15}{p}(r) (\log r) dr \right\},
\earr
$$
are the Cartesian components of the translation of the crack tip.
The second and third terms are of the order $O(r^{1/2})$, and the last term is a higher order term, which is of the order $O(r^{\beta}), \beta > 1/2$. 

It can be shown that $\beta = 1$, so that the computation of the residues leads to the asymptotics of the displacement field as follows
\beq
\lb{disp}
\bu(r,\theta) 
= \bV_0(\theta) + \bV_{1/2}(r,\theta) + O(r),
\eeq
where
$$
\bV_{1/2}(r,\theta) = 
\frac{i}{2\pi d_*}
\left\{-\left[\tilde{\bu}(s,\theta) \delta(s)\right]_{s = -1/2 + i\epsilon} 
r^{1/2 - i\epsilon}
+ \left[\tilde{\bu}(s,\theta) \delta(s)\right]_{s = -1/2 - i\epsilon} 
r^{1/2 + i\epsilon}\right\}.
$$

Note that the remainders $O(1)$ in \eq{sigma} and $O(r)$ in \eq{disp} are rough estimates and will be refined in the Section \ref{high}. However, their physical meaning is 
already evident: the first corresponds to the so-called T-stress and the second to a rigid body rotation superimposed to a uniform deformation near the crack tip. 

\subsubsection{The complex stress intensity factor $K = K_\modI + i K_\modII$}
From the full field solution outlined in Sections \ref{general} and \ref{full} we obtain
\beq
\lb{eq46}
\tilde{\sigma}_{\theta\theta}(s,0) + i \tilde{\sigma}_{r\theta}(s,0) = 
-\left\{\langle \tilde{p} \rangle(s) + i\langle \tilde{q} \rangle(s) + \frac{\alpha}{2} \jump{0.15}{\tilde{p}}(s) + i \frac{\alpha}{2} \jump{0.15}{\tilde{q}}(s)\right\}
\frac{\cos \pi s + i d_* \sin \pi s}{\delta(s)}.
\eeq

Applying the formula (\ref{sigma}), the leading term asymptotics may be written as
$$
\sigma_{\theta\theta}(r,0) + i \sigma_{r\theta}(r,0) = \frac{K}{\sqrt{2\pi}} r^{-1/2 + i\epsilon} + O(1),
$$
where
$$
K = - \sqrt{\frac{2}{\pi}} \cosh(\pi\epsilon) \left\{\langle \tilde{p} \rangle(s) + i\langle \tilde{q} \rangle(s) 
+ \frac{\alpha}{2} \jump{0.15}{\tilde{p}}(s) + i \frac{\alpha}{2} \jump{0.15}{\tilde{q}}(s)\right\}_{s = -1/2 - i\epsilon}
$$
is the complex stress intensity factor.

Correspondingly, by means of the formula (\ref{disp}), we obtain
$$
\jump{0.15}{u_\theta}(r) + i\jump{0.15}{u_r}(r) = 
-\frac{(1 - \nu_+) / \mu_+ + (1 - \nu_-) / \mu_-}{(1/2 + i\epsilon)\cosh(\pi\epsilon)} \frac{K}{\sqrt{2\pi}} r^{1/2 + i\epsilon} + O(r), 
$$
where we used the notation $\jump{0.15}{f}(r) = f(r,\pi) - f(r,-\pi)$. Note that the zero-order term $\bV_0(\theta)$ present in \eq{disp} disappears in the last formula.

The formula for the complex stress intensity factor is then
$$
\barr{l}
\ds
K = - \sqrt{\frac{2}{\pi}} \cosh(\pi\epsilon) 
\int_{0}^{\infty}\left\{\langle p \rangle(r) + i\langle q \rangle(r) + \frac{\alpha}{2} \jump{0.15}{p}(r) + i \frac{\alpha}{2} \jump{0.15}{q}(r)\right\} 
r^{-1/2 - i\epsilon} dr.
\earr
$$

As expected, this representation is consistent with Hutchinson, Mear and Rice (1987) who considered a crack loaded by point forces a distance $a$ behind the tip:
$$
p^+(r) = p^-(r) = -P \delta(r-a), \quad q^+(r) = q^-(r) = -Q \delta(r-a),
$$
where $\delta(\cdot)$ is the Dirac delta function, so that
$$
\langle p \rangle(r) = -P \delta(r-a), \quad \jump{0.15}{p}(r) = 0, \quad \langle q \rangle(r) = -Q \delta(r-a), \quad \jump{0.15}{q}(r) = 0.
$$
For this loading we obtain
$$
K  = \sqrt{\frac{2}{\pi}} \cosh(\pi\epsilon) \int_0^\infty \left\{ P \delta(r-a) + i Q \delta(r-a) \right\} r^{-1/2-i\epsilon} dr 
= \sqrt{\frac{2}{\pi}} \cosh(\pi\epsilon) (P+iQ) a^{-1/2-i\epsilon},
$$
which fully agrees with Hutchinson, Mear and Rice (1987).

\subsection{High-order asymptotics}
\lb{high}
The asymptotic procedure involving the evaluation of stress intensity factors for cracks with a slightly perturbed front requires the high-order asymptotics of the 
displacement and stress fields near the crack edge. The procedure of the previous section can be extended to the high-order terms, and hence the high-order asymptotics 
of the stress field is given by
\beq
\lb{beth}
\barr{l}
\bsigma(r,\theta) = \beth_{-1/2}(r,\theta) + \bT(\theta) + \beth_{1/2}(r,\theta)
+ \bS(\theta)r + \beth_{3/2}(r,\theta) + O(r^2), 
\earr
\eeq
where (the superscript $^\top$ denotes transposition, $x_1 = r \cos\theta$, $x_2 = r \sin\theta$, $\bQ(\theta)$ is given in \eq{qtheta})
$$
\bT(\theta) = \Res\left[\tilde{\bsigma}(s,\theta), s = -1\right] = T \bQ(\theta) \left[\barr{ll} 1 & 0 \\[3mm] 0 & 0 \earr\right] \bQ^\top(\theta), 
$$
$$
\barr{l}
\ds \bS(\theta) = \Res\left[\tilde{\bsigma}(s,\theta), s = -2\right] = \\[3mm]
\ds = \frac{1 \mp \alpha}{\pi \sqrt{x_1^2 + x_2^2}} \bQ(\theta) 
\left[
\barr{ll}
\ds x_2 \int_{0}^{\infty} \jump{0.15}{p}(r) \frac{dr}{r^2} - x_1 \int_{0}^{\infty} \jump{0.15}{q}(r) \frac{dr}{r^2} 
& \ds x_2 \int_{0}^{\infty} \jump{0.15}{q}(r) \frac{dr}{r^2} \\[3mm]
\ds x_2 \int_{0}^{\infty} \jump{0.15}{q}(r) \frac{dr}{r^2} 
& \ds 0 
\earr
\right] 
\bQ^\top(\theta),
\earr
$$
and
$$
\barr{ll}
\ds \beth_{\,l/2}(r,\theta) 
&
\ds = \sum_{\pm} \Res\left[\tilde{\bsigma}(s,\theta)\, r^{-s - 1}, s = -l/2 -1 \pm i\epsilon\right] \\[5mm]
&
\ds = \frac{i}{2\pi d_*}
\left\{-\left[\tilde{\bsigma}(s,\theta) \delta(s)\right]_{s = -l/2 -1 + i\epsilon} 
r^{l/2 - i\epsilon}
+ \left[\tilde{\bsigma}(s,\theta) \delta(s)\right]_{s = -l/2 -1 - i\epsilon} 
r^{l/2 + i\epsilon}\right\}, \quad l = 1,3.
\earr
$$

Note that $T$ is the T-stress given by
$$
T = \frac{1 \mp \alpha}{\pi} \int_{0}^{\infty} \jump{0.15}{q}(r) \frac{dr}{r}.
$$

Applying the formula \eq{beth} to \eq{eq46}, we finally obtain
$$
\sigma_{\theta\theta}(r,0) + i \sigma_{r\theta}(r,0) = \frac{K}{\sqrt{2\pi}} r^{-1/2 + i\epsilon} 
+ \frac{A}{\sqrt{2\pi}} r^{1/2 + i\epsilon} + \frac{B}{\sqrt{2\pi}} r^{3/2 + i\epsilon} 
+ O(r^{2}),
$$
where the constants in the high-order terms are 
$$
A = \sqrt{\frac{2}{\pi}} \cosh(\pi\epsilon) 
\int_0^\infty \left\{\langle p \rangle(r) + i\langle q \rangle(r) + \frac{\alpha}{2} \jump{0.15}{p}(r) + i \frac{\alpha}{2} \jump{0.15}{q}(r)\right\} 
r^{-3/2-i\epsilon} dr, 
$$
$$
B = -\sqrt{\frac{2}{\pi}} \cosh(\pi\epsilon) 
\int_0^\infty \left\{\langle p \rangle(r) + i\langle q \rangle(r) + \frac{\alpha}{2} \jump{0.15}{p}(r) + i \frac{\alpha}{2} \jump{0.15}{q}(r)\right\} 
r^{-5/2-i\epsilon} dr. 
$$
If the loading is smooth enough, then we can integrate by parts to obtain
$$
A = \sqrt{\frac{2}{\pi}} \frac{\cosh(\pi\epsilon)}{1/2+i\epsilon} 
\int_0^\infty \left\{\langle p \rangle'(r) + i\langle q \rangle'(r) + \frac{\alpha}{2} \jump{0.15}{p}'(r) + i \frac{\alpha}{2} \jump{0.15}{q}'(r)\right\} 
r^{-1/2-i\epsilon} dr,
$$
$$
B = -\sqrt{\frac{2}{\pi}} \frac{\cosh(\pi\epsilon)}{(1/2+i\epsilon)(3/2+i\epsilon)} 
\int_0^\infty \left\{\langle p \rangle''(r) + i\langle q \rangle''(r) + \frac{\alpha}{2} \jump{0.15}{p}''(r) + i \frac{\alpha}{2} \jump{0.15}{q}''(r)\right\} 
r^{-1/2-i\epsilon} dr.
$$

Correspondingly, we obtain for the high-order asymptotics of the displacement field the expression
$$
\bu(r,\theta) = \bV_0(\theta) + \bV_{1/2}(r,\theta) + \bV_1(\theta)r + \bV_{3/2}(r,\theta) 
+ \bV_2(\theta)r^2 + \bV_{5/2}(r,\theta) + O(r^3), 
$$
where
$$
\bV_l(\theta) = \Res\left[\tilde{\bu}(s,\theta), s = -l\right] = \bQ(\theta) r^{-l} \bw_l(r,\theta), \quad l = 1,2,
$$
\vspace{3mm}
$$
\barr{ll}
\ds \bV_{l/2}  
&
\ds = \sum_{\pm} \Res\left[\tilde{\bu}(s,\theta)\, r^{-s}, s = -l/2 \pm i\epsilon\right] \\[5mm]
&
\ds = \frac{i}{2\pi d_*}
\left\{-\left[\tilde{\bu}(s,\theta) \delta(s)\right]_{s = -l/2 + i\epsilon} 
r^{l/2 - i\epsilon}
+ \left[\tilde{\bu}(s,\theta) \delta(s)\right]_{s = -l/2 - i\epsilon} 
r^{l/2 + i\epsilon}\right\}, \quad l = 3,5,
\earr
$$
and
$$
w_{11} = \frac{1 \mp \alpha}{2\pi \mu_\pm} (1 - \nu_\pm) \left\{x_1\int_{0}^{\infty} \jump{0.15}{q}(r) \frac{dr}{r} - 
x_2\int_{0}^{\infty} \jump{0.15}{p}(r) \frac{dr}{r}\right\},
$$
$$
w_{12} = \frac{1 \mp \alpha}{2\pi \mu_\pm} \left\{(1 - \nu_\pm) x_1\int_{0}^{\infty} \jump{0.15}{p}(r) \frac{dr}{r} - 
\nu_\pm x_2\int_{0}^{\infty} \jump{0.15}{q}(r) \frac{dr}{r}\right\},
$$
$$
w_{21} = \frac{1 \mp \alpha}{8\pi \mu_\pm} \left\{4(1 - \nu_\pm)x_1x_2\int_{0}^{\infty} \jump{0.15}{p}(r) \frac{dr}{r^2} + 
[(x_1^2 + x_2^2) - (3 - 2\nu_\pm)(x_1^2 - x_2^2)]\int_{0}^{\infty} \jump{0.15}{q}(r) \frac{dr}{r^2}\right\},
$$
$$
w_{22} = \frac{1 \mp \alpha}{8\pi \mu_\pm} \left\{-[(x_1^2 + x_2^2) + (1 - 2\nu_\pm)(x_1^2 - x_2^2)]\int_{0}^{\infty} \jump{0.15}{p}(r) \frac{dr}{r^2} + 
4\nu_\pm x_1x_2 \int_{0}^{\infty} \jump{0.15}{q}(r) \frac{dr}{r^2}\right\}.
$$
We note that $\bw_l(r,\theta)$ is homogeneous of degree $l$ in the variable $r$ and that $w_{11}$ and $w_{12}$ are the Cartesian components of the local rigid body rotation superimposed to a uniform deformation near the crack tip mentioned above.

Hence
$$
\barr{l}
\ds \jump{0.15}{u_\theta}(r) + i\jump{0.15}{u_r}(r) = 
-\frac{(1 - \nu_+) / \mu_+ + (1 - \nu_-) / \mu_-}{\cosh(\pi\epsilon)} \times \\[3mm]
\ds \left\{ \frac{1}{1/2 + i\epsilon} \frac{K}{\sqrt{2\pi}} r^{1/2 + i\epsilon} 
- \frac{1}{3/2 + i\epsilon} \frac{A}{\sqrt{2\pi}} r^{3/2 + i\epsilon} + \frac{1}{5/2 + i\epsilon} \frac{B}{\sqrt{2\pi}} r^{5/2 + i\epsilon} \right\} + O(r^{3}). 
\earr
$$

\section{Symmetric and skew-symmetric weight functions for interfacial cracks}
\lb{wfunc}
In this section, we introduce a special type of singular solutions of a homogeneous problem for an interfacial crack. The traces of these functions on the plane containing 
the crack are known as the weight functions, and they are used in the evaluation of the stress intensity factors in models of linear fracture mechanics. The notations 
$\jump{0.15}{\bU}$ and $\langle \bU \rangle$ will be used for the symmetric and skew-symmetric weight function matrices, as outlined in the sequel of the paper.

We refer to the earlier publications by Antipov (1999), Bercial {\em et al.}\ (2005) and Piccolroaz {\em et al.}\ (2007) for the detailed discussion of the theory of 
weight functions and related functional equations of the Wiener-Hopf type. Here we give an outline, required for our purpose of evaluation of the weight function matrices 
for interfacial cracks.

The special singular solutions for the interfacial crack are defined as the solutions of the elasticity problem where the crack is placed along the positive semi-axis, 
$x_1 > 0$, the boundary conditions are homogeneous (traction-free crack faces) and satisfying special homogeneity properties. In particular, 
\bit
\item[(a)]
the singular solution $\bU = [U_1,U_2]^\top$ satisfies the equation of equilibrium;
\item[(b)]
$\jump{0.15}{\bU} = 0$ when $x_1 < 0$;
\item[(c)]
the associated traction vector acting on the plane containing the crack, $\bSigma = [\Sigma_{21},\Sigma_{22}]^\top$, is continuous and $\bSigma = 0$ when $x_2 = 0$ and 
$x_1 > 0$ (homogeneous boundary conditions); 
\item[(d)]
$\bU$ is a linear combination of homogeneous functions of degree $-1/2 + i\epsilon$ and $-1/2 - i\epsilon$.
\eit

We enphasise here that the domain for weight functions (where the crack is placed along the positive semi-axis) is different from the domain for the physical solution 
(where the crack is placed along the negative semi-axis). The reason for this is to have the fundamental integral identity, introduced in the sequel, eq. \eq{betti}, 
in convolution form (see Piccolroaz {\em et al.}, 2007, for details).

\subsection{The Wiener-Hopf equation}
\subsubsection{Factorization and solution}
Let us introduce the Fourier transforms of the crack-opening singular displacements and the corresponding traction components
$$
\jump{0.15}{\overline{\bU}}^+(\beta) = \int_{0}^{\infty} \jump{0.15}{\bU}(x_1) e^{i\beta x_1} dx_1, \quad
\overline{\bSigma}^-(\beta) = \int_{-\infty}^{0} \bSigma(x_1) e^{i\beta x_1} dx_1.
$$
The superscript $^+$ indicates that the function $\jump{0.15}{\overline{\bU}}^+(\beta)$ is analytic in the upper half-plane 
$\Complex^+ = \{\beta \in \Complex : \Re\beta \in (-\infty,\infty), \Im\beta \in (0,\infty)\}$, whereas the superscript $^-$ indicates that the function 
$\overline{\bSigma}^-(\beta)$ is analytic in the lower half-plane $\Complex^- = \{\beta \in \Complex : \Re\beta \in (-\infty,\infty), \Im\beta \in (-\infty,0)\}$.

These functions are related via the functional equation of the Wiener-Hopf type (Antipov, 1999):
\beq
\lb{wh}
\jump{0.15}{\overline{\bU}}^+(\beta) = 
-\frac{b}{|\beta|} \bG(\beta) \overline{\bSigma}^-(\beta), \quad \beta \in \Reals,
\eeq
where $b = (1 - \nu_+)/\mu_+ + (1 - \nu_-)/\mu_-$ and 
$$
\bG(\beta) = 
\left[
\barr{cc}
1 & -i\sign(\beta) d_* \\[3mm]
i\sign(\beta) d_* & 1
\earr
\right].
$$

The factorization of $-b/|\beta| \bG$ is given by
\beq
\lb{fact}
-\frac{b}{|\beta|} \bG(\beta) = 
-\frac{b}{\beta_+^{1/2} \beta_-^{1/2}} \bX^+(\beta) [\bX^-(\beta)]^{-1}, \quad \beta \in \Reals,
\eeq
where
$$
\bX^+(\beta) = d_0
\left[
\barr{cc}
\cos B^+ & -\sin B^+ \\[3mm]
\sin B^+ & \cos B^+
\earr
\right], \quad \beta \in \Complex^+, \quad
\bX^-(\beta) = d_0^{-1}
\left[
\barr{cc}
\cos B^- & -\sin B^- \\[3mm]
\sin B^- & \cos B^-
\earr
\right], \quad \beta \in \Complex^-,
$$
$d_0 = (1 - d_*)^{1/4}$, and the limit values of the functions $\beta_\pm^{1/2}$, as $\Im(\beta) \to 0^\pm$, are
$$
\beta_+^{1/2} =
\left\{ 
\barr{ll}
\beta^{1/2}, & \beta>0 \\[3mm]
i(-\beta)^{1/2}, & \beta<0
\earr
\right., \quad
\beta_-^{1/2} =
\left\{ 
\barr{ll}
\beta^{1/2}, & \beta>0 \\[3mm]
-i(-\beta)^{1/2}, & \beta<0
\earr
\right.,
$$
and
$$
B^\pm = -\epsilon \log(\mp i\beta), \quad \beta \in \Complex^\pm.
$$
From the definition of singular solutions given in the introduction of this section, it follows that the asymptotic behaviour of the weight functions is given by
$$
\barr{ll}
\jump{0.15}{\bU}(x_1) \sim x_1^{-1/2} \bF_1(x_1), & x_1 \to 0^+, \\[3mm]
\bSigma(x_1) \sim (-x_1)^{-3/2} \bF_2(x_1), & x_1 \to 0^-,
\earr
$$ 
$$
\barr{ll}
\jump{0.15}{\bU}(x_1) \sim x_1^{-1/2} \bF_3(x_1), & x_1 \to \infty, \\[3mm]
\bSigma(x_1) \sim (-x_1)^{-3/2} \bF_4(x_1), & x_1 \to -\infty,
\earr
$$
where $\bF_1,\bF_2,\bF_3,\bF_4$ are bounded functions. The application of the Abelian type theorem \ref{abel} and the Tauberian type theorem \ref{tau} 
(see Appendix \ref{app3}) gives
$$
\barr{ll}
\jump{0.15}{\overline{\bU}}^+(\beta) \sim \beta_+^{-1/2} \tilde{\bF}_1(\beta), & 
\beta \in \Complex^+, \quad \beta \to \infty, \\[3mm]
\overline{\bSigma}^-(\beta) \sim \beta_-^{1/2} \tilde{\bF}_2(\beta), & 
\beta \in \Complex^-, \quad \beta \to \infty,
\earr
$$ 
$$
\barr{ll}
\jump{0.15}{\overline{\bU}}^+(\beta) \sim \beta_+^{-1/2} \tilde{\bF}_3(\beta), & 
\beta \in \Complex^+, \quad \beta \to 0, \\[3mm]
\overline{\bSigma}(\beta) \sim \beta_-^{1/2} \tilde{\bF}_4(\beta), & 
\beta \in \Complex^-, \quad \beta \to 0,
\earr
$$
where $\tilde{\bF}_1,\tilde{\bF}_2,\tilde{\bF}_3,\tilde{\bF}_4$ are bounded functions. Substituting the representation \eq{fact} for the matrix $-b/|\beta|\bG(\beta)$ 
into the Wiener-Hopf equation \eq{wh}, we obtain
\beq
\lb{wiener}
\beta_+^{1/2} [\bX^+(\beta)]^{-1} \jump{0.15}{\overline{\bU}}^+(\beta) =
-\frac{b}{\beta_-^{1/2}} [\bX^-(\beta)]^{-1} \overline{\bSigma}^-(\beta), \quad \beta \in \Reals.
\eeq
Since $[\bX^+(\beta)]^{-1}$ and $[\bX^-(\beta)]^{-1}$ are bounded functions and taking into account the asymptotic behaviour of the weight functions at infinity, it 
follows that the LHS and RHS of \eq{wiener} behave as $O(1)$, $\beta \to \infty$, and from the Liouville theorem, they are equal to the same constant 
$\be = [e_1,e_2]^\top$, where $e_1,e_2$ are arbitrary, so that
\beq
\lb{wf}
\jump{0.15}{\overline{\bU}}^+(\beta) = \frac{1}{\beta_+^{1/2}} \bX^+(\beta)
\left[ 
\barr{l}
e_1 \\[3mm]
e_2
\earr
\right], \quad \beta \in \Complex^+, \quad
\overline{\bSigma}^-(\beta) = -\frac{\beta_-^{1/2}}{b} \bX^-(\beta)
\left[ 
\barr{l}
e_1 \\[3mm]
e_2
\earr
\right], \quad \beta \in \Complex^-.
\eeq

\subsubsection{The basis in the space of singular solutions}
From now on, we shall use the notion of weight functions, defined as traces of the singular displacement fields obtained in the above section. It is evident from the 
solution obtained above that the space of singular solutions is a 2-dimensional linear space. Two linearly independent weight functions can be selected as follows. 
Setting $e_1 = 1$, $e_2 = 0$ in \eq{wf}, we get the first weight function as:
\beq
\lb{first}
\barr{l}
\ds \jump{0.15}{\overline{U}_1^1}^+ = \frac{d_0}{\beta_+^{1/2}} \cos B^+ =
\frac{d_0}{2} \beta_+^{-1/2} \left(e_0 \beta_+^{i\epsilon} + \frac{1}{e_0} \beta_+^{-i\epsilon}\right), \quad 
\beta \in \Complex^+, \\[3mm]
\ds \jump{0.15}{\overline{U}_2^1}^+ = \frac{d_0}{\beta_+^{1/2}} \sin B^+ =
\frac{i d_0}{2} \beta_+^{-1/2} \left(e_0 \beta_+^{i\epsilon} - \frac{1}{e_0} \beta_+^{-i\epsilon}\right), \quad 
\beta \in \Complex^+, \\[3mm]
\ds \overline{\bSigma}_{21}^{1-} = -\frac{\beta_-^{1/2}}{bd_0} \cos B^- =
- \frac{1}{2bd_0} \beta_-^{1/2} \left(\frac{1}{e_0} \beta_-^{i\epsilon} + e_0 \beta_-^{-i\epsilon}\right), \quad 
\beta \in \Complex^-, \\[3mm]
\ds \overline{\bSigma}_{22}^{1-} = -\frac{\beta_-^{1/2}}{bd_0} \sin B^- =
- \frac{i}{2bd_0} \beta_-^{1/2} \left(\frac{1}{e_0} \beta_-^{i\epsilon} - e_0 \beta_-^{-i\epsilon}\right), \quad 
\beta \in \Complex^-,
\earr
\eeq
where $e_0 = e^{\epsilon\pi/2}$.

Setting $e_1 = 0$, $e_2 = 1$ in \eq{wf}, we get the second weight function as:
\beq
\lb{second}
\barr{l}
\ds \jump{0.15}{\overline{U}_1^2}^+ = -\frac{d_0}{\beta_+^{1/2}} \sin B^+ =
-\frac{i d_0}{2} \beta_+^{-1/2} \left(e_0 \beta_+^{i\epsilon} - \frac{1}{e_0} \beta_+^{-i\epsilon}\right), \quad 
\beta \in \Complex^+, \\[3mm]
\ds \jump{0.15}{\overline{U}_2^2}^+ = \frac{d_0}{\beta_+^{1/2}} \cos B^+ =
\frac{d_0}{2} \beta_+^{-1/2} \left(e_0 \beta_+^{i\epsilon} + \frac{1}{e_0} \beta_+^{-i\epsilon}\right), \quad 
\beta \in \Complex^+, \\[3mm]
\ds \overline{\bSigma}_{21}^{2-} = \frac{\beta_-^{1/2}}{bd_0} \sin B^- =
\frac{i}{2bd_0} \beta_-^{1/2} \left(\frac{1}{e_0} \beta_-^{i\epsilon} - e_0 \beta_-^{-i\epsilon}\right), \quad 
\beta \in \Complex^-, \\[3mm]
\ds \overline{\bSigma}_{22}^{2-} = -\frac{\beta_-^{1/2}}{bd_0} \cos B^- =
- \frac{1}{2bd_0} \beta_-^{1/2} \left(\frac{1}{e_0} \beta_-^{i\epsilon} + e_0 \beta_-^{-i\epsilon}\right), \quad 
\beta \in \Complex^-.
\earr
\eeq
In order to use a compact notation in the sequel of the paper, we collect the weight function components together with components of the corresponding tractions in 
matrices as follows
$$
\jump{0.15}{\overline{\bU}}^+ = 
\left[
\barr{cc}
\jump{0.15}{\overline{U}_1^1}^+ & \jump{0.15}{\overline{U}_1^2}^+ \\[3mm]
\jump{0.15}{\overline{U}_2^1}^+ & \jump{0.15}{\overline{U}_2^2}^+ 
\earr
\right], \quad
\overline{\bSigma}^- = 
\left[
\barr{cc}
\overline{\Sigma}_{21}^{1-} & \overline{\Sigma}_{21}^{2-} \\[3mm]
\overline{\Sigma}_{22}^{1-} & \overline{\Sigma}_{22}^{2-} 
\earr
\right].
$$

Note that the limit values of the functions $\beta_+^{\pm i\epsilon}$, as $\Im(\beta) \to 0^+$, are
$$
\beta_+^{i\epsilon} = 
\left\{ 
\barr{ll} 
\ds \beta^{i\epsilon}, & \beta>0, \\[3mm]
\ds \frac{(-\beta)^{i\epsilon}}{e_0^2}, & \beta<0,
\earr 
\right. \qquad
\beta_+^{-i\epsilon} = 
\left\{ 
\barr{ll} 
\ds \beta^{-i\epsilon}, & \beta>0, \\[3mm]
\ds (-\beta)^{-i\epsilon} e_0^2, & \beta<0,
\earr 
\right.
$$
whereas the limit values of the functions $\beta_-^{\pm i\epsilon}$, as $\Im(\beta) \to 0^-$, are
$$
\beta_-^{i\epsilon} = 
\left\{ 
\barr{ll} 
\ds \beta^{i\epsilon}, & \beta>0, \\[3mm]
\ds (-\beta)^{i\epsilon} e_0^2, & \beta<0,
\earr 
\right. \qquad
\beta_-^{-i\epsilon} = 
\left\{ 
\barr{ll} 
\ds \beta^{-i\epsilon}, & \beta>0, \\[3mm]
\ds \frac{(-\beta)^{-i\epsilon}}{e_0^2}, & \beta<0.
\earr 
\right.
$$

It is clear that the choice of a basis of two linearly independent singular solutions is not unique. Bercial-Velez {\em et al.}\ (2005) provided another set 
of linearly independent weight functions, namely
$$
\barr{l}
\ds \jump{0.15}{\overline{\mU}_1^1}^+ =  
\frac{1}{2d_0^2} \beta_+^{-1/2} 
\left(-\frac{e_0 \beta_+^{i\epsilon}}{c_1^-} + \frac{\beta_+^{-i\epsilon}}{e_0 c_1^+}\right), \quad 
\beta \in \Complex^+, \\[3mm]
\ds \jump{0.15}{\overline{\mU}_2^1}^+ =  
-\frac{i}{2d_0^2} \beta_+^{-1/2} 
\left(\frac{e_0 \beta_+^{i\epsilon}}{c_1^-} + \frac{\beta_+^{-i\epsilon}}{e_0 c_1^+}\right), \quad 
\beta \in \Complex^+, \\[3mm]
\ds \jump{0.15}{\overline{\mU}_1^2}^+ = -\jump{0.15}{\overline{\mU}_2^1}^+, \quad \beta \in \Complex^+, \\[3mm]
\ds \jump{0.15}{\overline{\mU}_2^2}^+ = \jump{0.15}{\overline{\mU}_1^1}^+, \quad \beta \in \Complex^+.
\earr
$$
It can be shown that these weight functions are linear combinations of (\ref{first}) and (\ref{second}),
$$
\left[
\barr{cc}
\jump{0.15}{\overline{\mU}_1^1}^+ & \jump{0.15}{\overline{\mU}_1^2}^+ \\[3mm]
\jump{0.15}{\overline{\mU}_2^1}^+ & \jump{0.15}{\overline{\mU}_2^2}^+ 
\earr
\right] = \frac{1}{2 c_1^+ c_1^- d_0^3}
\left[
\barr{cc}
\jump{0.15}{\overline{U}_1^1}^+ & \jump{0.15}{\overline{U}_1^2}^+ \\[3mm]
\jump{0.15}{\overline{U}_2^1}^+ & \jump{0.15}{\overline{U}_2^2}^+ 
\earr
\right] 
\left[
\barr{cc}
\ds -c_1^+ + c_1^- & \ds i (c_1^+ + c_1^-) \\[3mm]
\ds -i (c_1^+ + c_1^-) & \ds -c_1^+ + c_1^- 
\earr
\right].
$$

\subsection{The half-plane problem and full representation of weight functions}
In this section, we will construct the full representation of singular solutions for the whole plane. This is needed in order to compute the skew-symmetric weight 
function matrix $\langle \bU \rangle$. The complete singular solutions can be constructed by solving a boundary value problem for a semi-infinite half-plane subjected 
to traction boundary conditions at its boundary. 

Let us consider first the lower half-plane. A plane strain elasticity problem is usually solved by means of the Airy function. Introducing the Fourier transform with 
respect to the $x_1$ variable, the problem is most easily solved considering the stress component $\overline{\sigma}_{22}$ as the primary unknown function, so that the 
problem reduces to the following ordinary differential equation
$$
\overline{\sigma}_{22}'''' - 2\beta^2 \overline{\sigma}_{22}'' + \beta^4 \overline{\sigma}_{22} = 0,
$$
where a prime denotes derivative with respect to $x_2$. The general solution is then 
$$
\overline{\sigma}_{22}(\beta,x_2) = (A_2 + x_2 B_2) e^{|\beta| x_2}, \quad
\overline{\sigma}_{11}(\beta,x_2) = -\frac{1}{\beta^2} \overline{\sigma}_{22}'', \quad
\overline{\sigma}_{21}(\beta,x_2) = -\frac{i}{\beta} \overline{\sigma}_{22}', \quad \beta \in \Reals.
$$
Correspondingly, the Fourier transforms of the displacement components are 
$$
\overline{u}_{1}(\beta,x_2) = -\frac{i}{2\mu_- \beta} 
\{ \overline{\sigma}_{22} - (1-\nu_-)\overline{\sigma}_0 \}, \quad
\overline{u}_{2}(\beta,x_2) = \frac{1}{2\mu_- \beta^2} 
\{ \overline{\sigma}_{22}' + (1-\nu_-)\overline{\sigma}_0' \}, \quad \beta \in \Reals,
$$
where $\overline{\sigma}_0 = \overline{\sigma}_{11} + \overline{\sigma}_{22}$. The boundary conditions along the boundary $x_2 = 0^-$ are defined by
$$
\overline{\sigma}_{22}(\beta,x_2 = 0^-) = \overline{\Sigma}_{22}^-(\beta), \quad
\overline{\sigma}_{21}(\beta,x_2 = 0^-) = \overline{\Sigma}_{21}^-(\beta), \quad \beta \in \Reals,
$$
where $\Sigma_{22}^-,\Sigma_{21}^-$ are the tractions along the interface derived in the previous section. It follows that
$$
\overline{\sigma}_{22}(\beta,x_2 = 0^-) = A_2 = \overline{\Sigma}_{22}^-(\beta), \quad
\overline{\sigma}_{21}(\beta,x_2 = 0^-) = -\frac{i}{\beta} (A_2|\beta| + B_2) = 
\overline{\Sigma}_{21}^-(\beta), \quad \beta \in \Reals,
$$
and thus
$$
A_2 = \overline{\Sigma}_{22}^-, \quad
B_2 = i\beta \overline{\Sigma}_{21}^- - |\beta| \overline{\Sigma}_{22}^-, \quad \beta \in \Reals.
$$
The full representations of Fourier transforms (with respect to $x_1$)  of the required singular displacements and the corresponding components of stress are
$$
\barr{l}
\overline{\sigma}_{22}(\beta,x_2) = \{ i\beta x_2 \overline{\Sigma}_{21}^- 
+ (1 - |\beta|x_2) \overline{\Sigma}_{22}^- \} e^{|\beta|x_2}, \\[3mm]
\overline{\sigma}_{11}(\beta,x_2) = \{ -i(2\sign(\beta) + \beta x_2) \overline{\Sigma}_{21}^- 
+ (1 + |\beta|x_2) \overline{\Sigma}_{22}^- \} e^{|\beta|x_2}, \\[3mm]
\overline{\sigma}_{21}(\beta,x_2) = \{ (1 + |\beta|x_2) \overline{\Sigma}_{21}^- 
+ i\beta x_2 \overline{\Sigma}_{22}^- \} e^{|\beta|x_2}.
\earr
$$
$$
\barr{l}
\ds
\overline{u}_{1}(\beta,x_2) = \frac{1}{2\mu_-} \left\{ \left[x_2 + \frac{2(1-\nu_-)}{|\beta|}\right] 
\overline{\Sigma}_{21}^- 
+ i\left[\sign(\beta)x_2 + \frac{1 - 2\nu_-}{\beta}\right] 
\overline{\Sigma}_{22}^- \right\} e^{|\beta|x_2}, \\[3mm]
\ds
\overline{u}_{2}(\beta,x_2) = \frac{1}{2\mu_-} 
\left\{ i\left[\sign(\beta) x_2 - \frac{1-2\nu_-}{\beta} \right] 
\overline{\Sigma}_{21}^- 
+ \left[\frac{2(1-\nu_-)}{|\beta|} - x_2\right] 
\overline{\Sigma}_{22}^- \right\} e^{|\beta|x_2}.
\earr
$$

For the upper half-plane, we find the same equations, subject to replacing $|\beta|$ with $-|\beta|$, $\mu_-$ with $\mu_+$ and $\nu_-$ with $\nu_+$.

\subsection{New results for skew-symmetric weight functions}
\subsubsection{The Fourier transform representations}
It is possible now to derive the jump and average of Fourier transforms of the singular displacement functions across the plane containing the crack. The traces on the 
plane containing the crack are given by
$$
\barr{l}
\overline{U}_{1}(\beta) = \overline{u}_{1}(\beta,x_2 = 0^+) = 
\left[
\barr{ll}
\ds -\frac{1-\nu_+}{\mu_+ |\beta|}, & \ds \frac{i(1-2\nu_+)}{2\mu_+ \beta}
\earr
\right]
\left[
\barr{l}
\overline{\Sigma}_{21}^- \\[3mm]
\overline{\Sigma}_{22}^-
\earr
\right], \\[3mm]
\overline{U}_{2}(\beta) = \overline{u}_{2}(\beta,x_2 = 0^+) = 
\left[
\barr{ll}
\ds -\frac{i(1-2\nu_+)}{2\mu_+ \beta}, & \ds -\frac{1-\nu_+}{\mu_+ |\beta|}
\earr
\right]
\left[
\barr{l}
\overline{\Sigma}_{21}^- \\[3mm]
\overline{\Sigma}_{22}^-
\earr
\right], \\[3mm]
\overline{U}_{1}(\beta) = \overline{u}_{1}(\beta,x_2 = 0^-) = 
\left[
\barr{ll}
\ds \frac{1-\nu_-}{\mu_- |\beta|}, & \ds \frac{i(1-2\nu_-)}{2\mu_- \beta}
\earr
\right]
\left[
\barr{l}
\overline{\Sigma}_{21}^- \\[3mm]
\overline{\Sigma}_{22}^-
\earr
\right], \\[3mm]
\overline{U}_{2}(\beta) = \overline{u}_{2}(\beta,x_2 = 0^-) = 
\left[
\barr{ll}
\ds -\frac{i(1-2\nu_-)}{2\mu_- \beta}, & \ds \frac{1-\nu_-}{\mu_- |\beta|}
\earr
\right]
\left[
\barr{l}
\overline{\Sigma}_{21}^- \\[3mm]
\overline{\Sigma}_{22}^-
\earr
\right],
\earr
$$
so that, we obtain in matrix form
$$
\left[
\barr{cc}
\jump{0.15}{\overline{U}_1^1}^+ & \jump{0.15}{\overline{U}_1^2}^+ \\[3mm]
\jump{0.15}{\overline{U}_2^1}^+ & \jump{0.15}{\overline{U}_2^2}^+
\earr
\right] = -\frac{b}{|\beta|} 
\left[
\barr{cc}
1 & -i\sign(\beta) d_* \\[3mm]
i\sign(\beta) d_* & 1
\earr
\right]
\left[
\barr{cc}
\overline{\Sigma}_{21}^{1-} & \overline{\Sigma}_{21}^{2-} \\[3mm]
\overline{\Sigma}_{22}^{1-} & \overline{\Sigma}_{22}^{2-}
\earr
\right], \quad \beta \in \Reals,
$$
and
$$
\left[
\barr{cc}
\langle\overline{U}_1^1\rangle & \langle\overline{U}_1^2\rangle \\[3mm]
\langle\overline{U}_2^1\rangle & \langle\overline{U}_2^2\rangle
\earr
\right] = -\frac{b\alpha}{2|\beta|} 
\left[
\barr{cc}
1 & -i\sign(\beta) \gamma_* \\[3mm]
i\sign(\beta) \gamma_* & 1
\earr
\right]
\left[
\barr{cc}
\overline{\Sigma}_{21}^{1-} & \overline{\Sigma}_{21}^{2-} \\[3mm]
\overline{\Sigma}_{22}^{1-} & \overline{\Sigma}_{22}^{2-}
\earr
\right], \quad \beta \in \Reals,
$$
where $\gamma_*$ is a material parameter,
$$
\gamma_* = \frac{\mu_-(1 - 2\nu_+) + \mu_+(1 - 2\nu_-)}{2\mu_-(1 - \nu_+) - 2\mu_+(1 - \nu_-)}.
$$

Note that the above equations can be rewritten as
$$
\jump{0.15}{\overline{\bU}}^+ = 
-\frac{b}{|\beta|} \left[\barr{cc} 1 & 0 \\[3mm] 0 & 1 \earr\right] \overline{\bSigma}^- - 
\frac{i b d_*}{\beta} \left[\barr{cc} 0 & -1 \\[3mm] 1 & 0 \earr\right] \overline{\bSigma}^-, \quad
\beta \in \Reals,
$$
$$
\langle\overline{\bU}\rangle = 
-\frac{\alpha b}{2|\beta|} \left[\barr{cc} 1 & 0 \\[3mm] 0 & 1 \earr\right] \overline{\bSigma}^- - 
\frac{i \alpha b \gamma_*}{2\beta} \left[\barr{cc} 0 & -1 \\[3mm] 1 & 0 \earr\right]
\overline{\bSigma}^-, \quad \beta \in \Reals,
$$
so that we can derive the decomposition of $\langle\overline{\bU}\rangle$ in the sum of "$+$" and "$-$" functions as follows
\beq
\lb{average}
\langle\overline{\bU}\rangle = 
\frac{\alpha}{2} \jump{0.15}{\overline{\bU}}^+ + \alpha(d_* - \gamma_*) 
\frac{i b}{2\beta} \left[\barr{cc} 0 & -1 \\[3mm] 1 & 0 \earr\right] 
\overline{\bSigma}^-, \quad \beta \in \Reals,
\eeq
where the last term on the right-hand side is a "$-$" function.

\subsubsection{The weight functions - Fourier inversion}
\lb{weight_f}
After the inversion of the corresponding Fourier transforms, the weight functions, which will be needed for the computation of stress intensity factors, are as follows. 
The symmetric weight function matrix $\jump{0.15}{\bU}(x_1)$ is equal to 0 for $x_1 < 0$, whereas for $x_1 > 0$ it is given by 
\beq
\lb{jumpu}
\barr{l}
\ds
\jump{0.15}{U_1^1}(x_1) = \frac{x_1^{-1/2}}{2d_0 \sqrt{2\pi}} 
\left( \frac{x_1^{-i\epsilon}}{c_1^+} + \frac{x_1^{i\epsilon}}{c_1^-} \right), \\[5mm]
\ds
\jump{0.15}{U_2^1}(x_1) = \frac{i x_1^{-1/2}}{2d_0 \sqrt{2\pi}} 
\left( \frac{x_1^{-i\epsilon}}{c_1^+} - \frac{x_1^{i\epsilon}}{c_1^-} \right), \\[5mm]
\ds 
\jump{0.15}{U_1^2}(x_1) = -\jump{0.15}{U_2^1}(x_1), \\[5mm]
\jump{0.15}{U_2^2}(x_1) = \jump{0.15}{U_1^1}(x_1).
\earr
\eeq
The skew-symmetric weight function matrix $\langle\bU\rangle(x_1)$ is equal to $\ds \frac{\alpha}{2}\jump{0.15}{\bU}(x_1)$ for $x_1 > 0$, whereas for $x_1 < 0$ it is 
given by 
\beq
\lb{meanu}
\barr{l}
\ds
\langle U_1^1 \rangle(x_1) = -\frac{i\alpha(d_*-\gamma_*)(-x_1)^{-1/2}}{4d_0^3 \sqrt{2\pi}}
\left[ \frac{(-x_1)^{-i\epsilon}}{c_1^+} - \frac{(-x_1)^{i\epsilon}}{c_1^-} \right], \\[3mm]
\ds
\langle U_2^1 \rangle(x_1) = \frac{\alpha(d_*-\gamma_*)(-x_1)^{-1/2}}{4d_0^3 \sqrt{2\pi}}
\left[ \frac{(-x_1)^{-i\epsilon}}{c_1^+} + \frac{(-x_1)^{i\epsilon}}{c_1^-} \right], \\[3mm]
\ds 
\langle U_1^2 \rangle(x_1) = -\langle U_2^1 \rangle(x_1), \\[3mm]
\langle U_2^2 \rangle(x_1) = \langle U_1^1 \rangle(x_1).
\earr
\eeq
The function $\bSigma(x_1)$ is equal to 0 for $x_1 > 0$, whereas for $x_1 < 0$ it is given by 
$$
\barr{l}
\ds
\Sigma_{21}^1(x_1) = \frac{(-x_1)^{-3/2}}{2bd_0^3 \sqrt{2\pi}} 
\left[ \frac{1/2 + i\epsilon}{c_1^+}(-x_1)^{-i\epsilon} + \frac{1/2 - i\epsilon}{c_1^-}(-x_1)^{i\epsilon} \right], \\[3mm]
\ds
\Sigma_{22}^1(x_1) = \frac{i(-x_1)^{-3/2}}{2bd_0^3 \sqrt{2\pi}} 
\left[ \frac{1/2 + i\epsilon}{c_1^+}(-x_1)^{-i\epsilon} - \frac{1/2 - i\epsilon}{c_1^-}(-x_1)^{i\epsilon} \right], \\[3mm]
\ds 
\Sigma_{21}^2(x_1) = -\Sigma_{22}^1(x_1), \\[3mm]
\Sigma_{22}^2(x_1) = \Sigma_{21}^1(x_1).
\earr
$$

It will be shown in the sequel of the text that we only need the inverse transform of $\jump{0.15}{\overline{\bU}}^+$  (see (\ref{average})) in 
order to compute the stress intensity factors for both symmetric and skew-symmetric loading.

\section{The Betti identity and evaluation of the stress intensity factors}
\lb{identity}
Here we develop a general procedure for the evaluation of coefficients in the asymptotics of the stress components near the crack tip. This includes the stress intensity 
factors as well as high-order asymptotics. In particular, the coefficients near the higher order terms require appropriate weight functions which are shown to be derived 
via differentiation of the weight functions $\jump{0.15}{\bU}$ and $\langle \bU \rangle$ along the crack.

\subsection{The fundamental Betti identity and its equivalent representation in terms of Fourier transforms}
In this section, the notations $\bu = [u_1, u_2]^\top$ and $\bsigma = [\sigma_{21},\sigma_{22}]^\top$ will be used for the physical displacement and traction fields discussed 
in Section \ref{prel_analysis}. The notations $\bU = [U_1,U_2]^\top$ and $\bSigma = [\Sigma_{21},\Sigma_{22}]^\top$ will apply to the auxiliary singular displacements and the 
corresponding tractions obtained in Section \ref{weight_f}. We note that $\bU$ is discontinuous along the positive semi-axis $x_1 > 0$, whereas $\bu$ is discontinuous 
along the negative semi-axis $x_1 < 0$.

Similar to Willis and Movchan (1995) and Piccolroaz {\em et al.}\ (2007), we can apply the Betti formula to the physical fields and to the weight functions in order to 
evaluate the coefficients in the asymptotics near the crack tip. In particular, applying the Betti identity to the upper half-plane and lower half-plane, we obtain
$$
\int_{(x_2 = 0^+)} \left\{ \bU^\top(x_1'-x_1,0^+) \bR \bsigma(x_1,0^+) - \bSigma^\top(x_1'-x_1,0^+) \bR \bu(x_1,0^+) \right\} dx_1 = 0,
$$
and
$$
\int_{(x_2 = 0^-)} \left\{ \bU^\top(x_1'-x_1,0^-) \bR \bsigma(x_1,0^-) - \bSigma^\top(x_1'-x_1,0^-) \bR \bu(x_1,0^-) \right\} dx_1 = 0,
$$
respectively, where 
\beq
\lb{erre}
\bR = \begin{pmatrix} -1 & 0 \\ 0 & 1 \end{pmatrix}
\eeq
is a rotation matrix. Then, subtracting one from the other, we arrive at
$$
\barr{l}
\ds
\int_{(x_2 = 0)} \left\{ \bU^\top(x_1'-x_1,0^+) \bR \bsigma(x_1,0^+) - \bU^\top(x_1'-x_1,0^-) \bR \bsigma(x_1,0^-) \right. \\[3mm]
\ds \hspace{30mm}
\left. - \left[\bSigma^\top(x_1'-x_1,0^+) \bR \bu(x_1,0^+) - \bSigma^\top(x_1'-x_1,0^-) \bR \bu(x_1,0^-)\right] \right\} dx_1 = 0.
\earr
$$
The traction components acting on the $x_1$ axis can be written as
$$
\bsigma(x_1,0^+) = \bp^+(x_1) + \bsigma^{(+)}(x_1,0), \quad \bsigma(x_1,0^-) = \bp^-(x_1) + \bsigma^{(+)}(x_1,0), 
$$
where $\bp^+(x_1) = \bsigma(x_1,0^+) H(-x_1)$, $\bp^-(x_1) = \bsigma(x_1,0^-) H(-x_1)$ is the loading acting on the upper and lower crack faces, respectively, and 
$\bsigma^{(+)}(x_1,0)$ is the traction field ahead of the crack tip, with $H(x_1)$ being the unit-step Heaviside function. Consequently, the integral identity can be 
written as
$$
\barr{l}
\ds
\int_{(x_2 = 0)} \Big\{ \bU^\top(x_1'-x_1,0^+) \bR \bp^+(x_1) + \bU^\top(x_1'-x_1,0^+) \bR \bsigma^{(+)}(x_1,0) \\[3mm]
\ds \hspace{15mm}
- \bU^\top(x_1'-x_1,0^-) \bR \bp^-(x_1) - \bU^\top(x_1'-x_1,0^-) \bR \bsigma^{(+)}(x_1,0) \\[3mm]
\ds \hspace{30mm}
- \left[\bSigma^\top(x_1'-x_1,0^+) \bR \bu(x_1,0^+) - \bSigma^\top(x_1'-x_1,0^-) \bR \bu(x_1,0^-)\right] \Big\} dx_1 = 0.
\earr
$$
From the continuity of $\bsigma^{(+)}$ and $\bSigma$ we get
$$
\barr{l}
\ds
\int_{(x_2 = 0)} \Big\{ \jump{0.15}{\bU}^\top(x_1'-x_1) \bR \bsigma^{(+)}(x_1,0) - \bSigma^\top(x_1'-x_1,0) \bR \jump{0.15}{\bu}(x_1) \Big\} dx_1 = \\[3mm]
\ds \hspace{30mm}
- \int_{(x_2 = 0)} \Big\{ \bU^\top(x_1'-x_1,0^+) \bR \bp^+(x_1) - \bU^\top(x_1'-x_1,0^-) \bR \bp^-(x_1) \Big\} dx_1.
\earr
$$
Introducing the symmetric and skew-symmetric parts of the loading (\ref{parts}), we finally obtain
\beq
\lb{betti}
\barr{l}
\ds
\int_{(x_2 = 0)} \Big\{ \jump{0.15}{\bU}^\top(x_1'-x_1) \bR \bsigma^{(+)}(x_1,0) - \bSigma^\top(x_1'-x_1,0) \bR \jump{0.15}{\bu}(x_1) \Big\} dx_1 = \\[3mm]
\ds \hspace{30mm}
- \int_{(x_2 = 0)} \Big\{ \jump{0.15}{\bU}^\top(x_1'-x_1) \bR \langle \bp \rangle(x_1) + \langle \bU \rangle^\top(x_1'-x_1) \bR \jump{0.15}{\bp}(x_1) \Big\} dx_1.
\earr
\eeq
Taking the Fourier transform of (\ref{betti}) we derive
\beq
\lb{betti2}
\jump{0.15}{\overline{\bU}}^{+\top} \bR \overline{\bsigma}^+ - \overline{\bSigma}^{-\top} \bR \jump{0.15}{\overline{\bu}}^- = 
- \jump{0.15}{\overline{\bU}}^{+\top} \bR \langle \overline{\bp} \rangle - \langle \overline{\bU} \rangle^\top \bR \jump{0.15}{\overline{\bp}}, \quad \beta \in \Reals.
\eeq

\subsection{Stress intensity factors and high-order terms coefficients}
In this section, we will derive the stress intensity factors as well as high-order terms from the Betti formula \eq{betti2} obtained in the previous section. The complex 
stress intensity factor will be denoted by $K = K_\modI + i K_\modII$, and analogous notations will be used for high-order terms: $A = A_\modI + i A_\modII$, 
$B = B_\modI + i B_\modII$.

The asymptotics of the physical traction field as $x_1 \to 0^+$ are as follows
$$
\bsigma^{(+)}(x_1) = \frac{x_1^{-1/2}}{2\sqrt{2\pi}} \bmS(x_1) \bK + \frac{x_1^{1/2}}{2\sqrt{2\pi}} \bmS(x_1) \bA + \frac{x_1^{3/2}}{2\sqrt{2\pi}} \bmS(x_1) \bB + O(x_1^{5/2}),
$$
where 
$$
\bmS(x_1) = 
\left[\barr{ll}
\ds
-i x_1^{i\epsilon} & i x_1^{-i\epsilon} \\[3mm]
\ds
x_1^{i\epsilon} & x_1^{-i\epsilon}
\earr\right],
$$
and $\bK = [K, K^*]^\top$, $\bA = [A, A^*]^\top$, $\bB = [B, B^*]^\top$, where the superscript $^*$ denotes conjugation. The corresponding Fourier transforms, as 
$\beta \to \infty$, $\beta \in \Complex^+$, are (compare with Piccolroaz {\em et al.}, 2007)
\beq
\lb{stress1}
\overline{\bsigma}^+(\beta) = \frac{\beta_+^{-1/2}}{4} \bmT_1(\beta) \bK + \frac{\beta_+^{-1/2}}{4\beta} \bmT_2(\beta) \bA + 
\frac{\beta_+^{-1/2}}{4\beta^2} \bmT_3(\beta) \bB + O(\beta^{-7/2}),
\eeq
where
\beq
\bmT_j(\beta) = 
\left[\barr{ll}
\ds \frac{\beta_+^{-i\epsilon}}{c_j^+ e_0} & \ds -\frac{e_0 \beta_+^{i\epsilon}}{c_j^-} \\[5mm]
\ds \frac{i\beta_+^{-i\epsilon}}{c_j^+ e_0}  & \ds \frac{ie_0 \beta_+^{i\epsilon}}{c_j^-} 
\earr\right], \quad j=1,2,3,
\eeq
\beq
\lb{stress2}
c_1^\pm = \frac{(1+i)\sqrt{\pi}}{2\Gamma(1/2 \pm i\epsilon)}, \quad
c_2^\pm = \frac{(1-i)\sqrt{\pi}}{2\Gamma(3/2 \pm i\epsilon)}, \quad
c_3^\pm = -\frac{(1+i)\sqrt{\pi}}{2\Gamma(5/2 \pm i\epsilon)},
\eeq
and $\Gamma(\cdot)$ denotes the Gamma function.

The components of the displacement physical field as $x_1 \to 0^-$ have the form
$$
\jump{0.15}{\bu}(x_1) = \frac{b d_0^2}{\sqrt{2\pi}} (-x_1)^{1/2} \bmU_1(x_1) \bK -
\frac{b d_0^2}{\sqrt{2\pi}} (-x_1)^{3/2} \bmU_2(x_1) \bA + \frac{b d_0^2}{\sqrt{2\pi}} (-x_1)^{5/2} \bmU_3(x_1) \bB + O[(-x_1)^{7/2}],
$$
where
$$
\bmU_j(x_1) = 
\left[\barr{ll}
\ds -\frac{i (-x_1)^{i\epsilon}}{2j-1 + 2i\epsilon} & 
\ds \frac{i (-x_1)^{-i\epsilon}}{2j-1 - 2i\epsilon} \\[3mm]
\ds \frac{(-x_1)^{i\epsilon}}{2j-1 + 2i\epsilon} & 
\ds \frac{(-x_1)^{-i\epsilon}}{2j-1 - 2i\epsilon}
\earr\right], \quad j=1,2,3.
$$

The corresponding Fourier transforms, as $\beta \to \infty$, $\beta \in \Complex^-$ are (compare with Piccolroaz {\em et al.}, 2007)
\beq
\lb{displac}
\jump{0.15}{\overline{\bu}}^-(\beta) = -\frac{bd_0^2}{4\beta} \beta_-^{-1/2} \bmV_1(\beta) \bK - \frac{bd_0^2}{4\beta^2} \beta_-^{-1/2} \bmV_2(\beta) \bA 
- \frac{bd_0^2}{4\beta^3} \beta_-^{-1/2} \bmV_3(\beta) \bB + O(\beta^{-9/2}),
\eeq
where
$$
\bmV_j(\beta) = 
\left[\barr{ll}
\ds \frac{e_0 \beta_-^{-i\epsilon}}{c_j^+} & \ds - \frac{\beta_-^{i\epsilon}}{c_j^-e_0} \\[5mm]
\ds \frac{i e_0 \beta_-^{-i\epsilon}}{c_j^+} & \ds \frac{i \beta_-^{i\epsilon}}{c_j^-e_0} 
\earr\right], \quad j=1,2,3.
$$

From these asymptotics we can derive estimates of the terms in the LHS of the Betti identity 
\eq{betti2}, as $\beta \to \infty$,
\beq
\lb{lhs}
\barr{l}
\jump{0.15}{\overline{\bU}}^{+\top} \bR \overline{\bsigma}^+  = 
\beta^{-1} \bmM_1 \bK + \beta^{-2} \bmM_2 \bA + \beta^{-3} \bmM_3 \bB + O(\beta^{-4}), \quad \beta \in \Complex^+ \\[3mm]
\overline{\bSigma}^{-\top} \bR \jump{0.15}{\overline{\bu}}^- = 
\beta^{-1} \bmM_1 \bK + \beta^{-2} \bmM_2 \bA + \beta^{-3} \bmM_3 \bB + O(\beta^{-4}), \quad \beta \in \Complex^-
\earr
\eeq
where
\beq
\lb{emme1}
\bmM_j = \frac{d_0}{4c_j^+c_j^-} \left[ \barr{cc} -c_j^- & c_j^+ \\[3mm] ic_j^- & ic_j^+ \earr \right], \quad j=1,2,3.
\eeq

Let us introduce the following notation for the RHS of the Betti identity \eq{betti2}
$$
\bPsi(\beta) = - \jump{0.15}{\overline{\bU}}^{+\top}(\beta) \bR \langle \overline{\bp} \rangle(\beta) 
- \langle\overline{\bU}\rangle^\top(\beta) \bR \jump{0.15}{\overline{\bp}}(\beta), \quad \beta \in \Reals.
$$
Note that for any possible loading, $\bPsi(\beta) \to 0$, as $\beta \to \pm \infty$. In fact, if the loading is given by point forces in terms of the Dirac delta function, 
then $\bPsi(\beta) = O(|\beta|^{-1/2})$. If, for example, $\bp^\pm \in L_1(\Reals)$ or $(\bp^\pm)' \in L_1(\Reals)$ then $\bPsi(\beta) = o(|\beta|^{-1/2})$ or 
$\bPsi(\beta) = o(|\beta|^{-3/2})$, respectively.

Consequently, we can split $\bPsi(\beta)$ in the sum of a plus function and a minus function 
$$
\bPsi(\beta) = \bPsi^+(\beta) - \bPsi^-(\beta), \quad \beta \in \Reals,
$$
where
$$
\bPsi^\pm(\beta) = \frac{1}{2\pi i} \int_{-\infty}^{\infty} \frac{\bPsi(t)}{t-\beta} dt, \quad 
\beta \in \Complex^\pm.
$$

Moreover $\bPsi^\pm(\beta) \sim 1/\beta$, $\beta \to \infty$, $\beta \in \Complex^\pm$, and, from the estimates (\ref{lhs}) we can conclude that the inhomogeneous 
Wiener-Hopf equation (\ref{betti2}) has the unique solution 
$$
\jump{0.15}{\overline{\bU}}^{+\top} \bR \overline{\bsigma}^+ = \bPsi^+, \quad \beta \in \Complex^+ \quad \text{and} \quad 
\overline{\bSigma}^{-\top} \bR \jump{0.15}{\overline{\bu}}^- = \bPsi^-, \quad \beta \in \Complex^-.
$$

From these identities we can extract asymptotic estimates. For this reason, let us first assume that $\bp^\pm \in \mS(\Reals)$, where $\mS(\Reals)$ is 
the Schwartz space of rapidly decreasing functions (Zemanian, 1968), so that 
$\bPsi(\beta)$ decays exponentially as $\beta \to \pm\infty$. Then we have, for $\beta \to \infty$, $\beta \in \Complex^\pm$,
\beq
\lb{rhs}
\bPsi^\pm(\beta) = \frac{1}{2\pi i} \int_{-\infty}^{\infty} \frac{\bPsi(t)}{t-\beta} dt = 
-\beta^{-1} \frac{1}{2\pi i} \int_{-\infty}^{\infty} \bPsi(t) dt  
-\beta^{-2} \frac{1}{2\pi i} \int_{-\infty}^{\infty} t \bPsi(t) dt
+ O(\beta^{-3}). 
\eeq

Comparing corresponding terms in (\ref{rhs}) and (\ref{lhs}), we obtain the following formulae for the complex stress intensity factor and the coefficient in the 
second-order term
\beq
\lb{coeff}
\barr{l}
\ds
\bK = \frac{1}{2\pi i} \bmM_1^{-1}
\int_{-\infty}^{\infty} \left\{\jump{0.15}{\overline{\bU}}^{+\top}(t) \bR \langle \overline{\bp} \rangle(t)  
+ \langle\overline{\bU}\rangle^\top(t) \bR \jump{0.15}{\overline{\bp}}(t) \right\} dt, \\[3mm]
\ds
\bA = \frac{1}{2\pi i} \bmM_2^{-1}
\int_{-\infty}^{\infty} t \left\{\jump{0.15}{\overline{\bU}}^{+\top}(t) \bR \langle \overline{\bp} \rangle(t) 
+ \langle\overline{\bU}\rangle^\top(t) \bR \jump{0.15}{\overline{\bp}}(t) \right\} dt,
\earr
\eeq
respectively.

Note that
$$
\barr{l}
\ds
\frac{1}{2\pi} \int_{-\infty}^{\infty} \left\{\jump{0.15}{\overline{\bU}}^{+\top}(t) \bR \langle \overline{\bp} \rangle(t)  
+ \langle\overline{\bU}\rangle^\top(t) \bR \jump{0.15}{\overline{\bp}}(t) \right\} dt = \\[3mm]
\ds
= \lim_{x_1'\to 0} \frac{1}{2\pi} 
\int_{-\infty}^{\infty} \left\{\jump{0.15}{\overline{\bU}}^{+\top}(t) \bR \langle \overline{\bp} \rangle(t) 
+ \langle\overline{\bU}\rangle^\top(t) \bR \jump{0.15}{\overline{\bp}}(t) \right\} e^{-i x_1' t} dt \\[3mm]
\ds
= \lim_{x_1'\to 0} \mF^{-1}_{x_1'} \left\{\jump{0.15}{\overline{\bU}}^{+\top} \bR \langle \overline{\bp} \rangle 
+ \langle\overline{\bU}\rangle^\top \bR \jump{0.15}{\overline{\bp}} \right\} \\[3mm]
\ds
= \lim_{x_1'\to 0} \int_{-\infty}^{0} \left\{\jump{0.15}{\bU}^\top(x_1'-x_1) \bR \langle \bp \rangle(x_1) 
+ \langle\bU\rangle^\top(x_1'-x_1) \bR \jump{0.15}{\bp}(x_1) \right\} dx_1.
\earr
$$

The formula for the complex stress intensity factor becomes
\beq
\lb{SIF}
\bK = -i \bmM_1^{-1} \lim_{x_1'\to 0} \int_{-\infty}^{0} \left\{\jump{0.15}{\bU}^\top(x_1'-x_1) \bR \langle \bp \rangle(x_1) 
+ \langle\bU\rangle^\top(x_1'-x_1) \bR \jump{0.15}{\bp}(x_1) \right\} dx_1.
\eeq

Similarly, we can develop the integral in \eq{coeff}$_2$ to obtain
$$
\barr{l}
\ds
\frac{1}{2\pi} \int_{-\infty}^{\infty} t \left\{\jump{0.15}{\overline{\bU}}^{+\top}(t) \bR \langle \overline{\bp} \rangle(t) 
+ \langle\overline{\bU}\rangle^\top(t) \bR \jump{0.15}{\overline{\bp}}(t) \right\} dt = \\[3mm]
\ds
= i \lim_{x_1'\to 0} \frac{d}{dx_1'} \frac{1}{2\pi} 
\int_{-\infty}^{\infty} \left\{\jump{0.15}{\overline{\bU}}^{+\top}(t) \bR \langle \overline{\bp} \rangle(t) 
+ \langle\overline{\bU}\rangle^\top(t) \bR \jump{0.15}{\overline{\bp}}(t) \right\} e^{-i x_1' t} dt \\[3mm]
\ds
= i \lim_{x_1'\to 0} \int_{-\infty}^{0} \left\{\jump{0.15}{\bU}^\top(x_1'-x_1) \bR \frac{d\langle \bp \rangle(x_1)}{dx_1} 
+ \langle\bU\rangle^\top(x_1'-x_1) \bR \frac{d\jump{0.15}{\bp}(x_1)}{dx_1} \right\} dx_1.
\earr
$$

The formula for the coefficient in the second order term is
\beq
\lb{sec}
\bA = \bmM_2^{-1} \lim_{x_1'\to 0} \int_{-\infty}^{0} \left\{\jump{0.15}{\bU}^\top(x_1'-x_1) \bR \frac{d\langle \bp \rangle(x_1)}{dx_1} 
+ \langle\bU\rangle^\top(x_1'-x_1) \bR \frac{d\jump{0.15}{\bp}(x_1)}{dx_1} \right\} dx_1.
\eeq

If the loading $\bp^\pm$ is not smooth enough, singular or it is a distribution (such as the Dirac delta function), then we can replace $\bp^\pm$ 
with a sequence $\{\bp^\pm_n\}$ such that $\bp^\pm_n \in \mS(\Reals)$ and $\bp^\pm_n \to \bp^\pm$ in $\bar{\mS}(\Reals)$, when $n \to \infty$. 
Here $\bar{\mS}(\Reals)$ denotes the space of distributions (Zemanian, 1968). Therefore, the representations \eq{coeff}, \eq{SIF}, \eq{sec} can 
be interpreted in the sense of distributions and we can integrate formally \eq{sec} by parts to obtain
$$
\bA = \bmM_2^{-1} \lim_{x_1'\to 0} \int_{-\infty}^{0} \left\{\frac{d\jump{0.15}{\bU}^\top(x_1'-x_1)}{dx_1} \bR \langle \bp \rangle(x_1) 
+ \frac{d\langle\bU\rangle^\top(x_1'-x_1)}{dx_1} \bR \jump{0.15}{\bp}(x_1) \right\} dx_1.
$$

\section{Perturbation of the crack front}
\lb{perturbation}
For the 2D case, the perturbation considered here is related to an advance of the crack front by a distance $a$. We denote quantities relative to the original unperturbed 
crack problem by subscript $0$, and quantities relative to the perturbed crack problem by subscript $\star$.

\subsection{General settings}
The complex stress intensity factor for the original unperturbed problem is given by
$$
\bK_0 = -i \bmM_1^{-1} \lim_{x_1'\to 0} \int_{-\infty}^{0} \left\{\jump{0.15}{\bU}^\top(x_1'-x_1) \bR \langle \bp \rangle(x_1) 
+ \langle\bU\rangle^\top(x_1'-x_1) \bR \jump{0.15}{\bp}(x_1) \right\} dx_1.
$$
For the perturbed problem, if the coordinate system is shifted by the quantity $a$ to the new crack tip position, the complex stress intensity factor is given by the 
formula
$$
\bK_\star(a) = -i \bmM_1^{-1} \lim_{x_1'\to 0} \int_{-\infty}^{0} \left\{\jump{0.15}{\bU}^\top(x_1'-x_1) \bR \langle \bp \rangle(x_1+a) 
+ \langle\bU\rangle^\top(x_1'-x_1) \bR \jump{0.15}{\bp}(x_1+a) \right\} dx_1.
$$

The computation of the first order variation of the complex stress intensity factor is now straightforward
$$
\barr{l}
\ds
\left.\frac{d\bK_\star(a)}{da}\right|_{a=0} =  \\[3mm]
\ds
=-i \bmM_1^{-1} \lim_{x_1'\to 0} \lim_{a \to 0^+} \int_{-\infty}^{0} 
\left\{\jump{0.15}{\bU}^\top(x_1'-x_1) \bR \frac{\langle \bp \rangle(x_1+a) - \langle \bp \rangle(x_1)}{a} \right. \\[3mm] 
\hspace{50mm}
\ds
\left. + \langle\bU\rangle^\top(x_1'-x_1) \bR \frac{\jump{0.15}{\bp}(x_1+a) - \jump{0.15}{\bp}(x_1)}{a} \right\} 
dx_1 \\[3mm]
\ds
=-i \bmM_1^{-1} \lim_{x_1'\to 0} \int_{-\infty}^{0} 
\left\{\jump{0.15}{\bU}^\top(x_1'-x_1) \bR \frac{d\langle \bp \rangle(x_1)}{dx_1} 
+ \langle\bU\rangle^\top(x_1'-x_1) \bR \frac{d\jump{0.15}{\bp}(x_1)}{dx_1} \right\} 
dx_1.
\earr
$$

Similarly, for the second order term, we obtain
$$
\left.\frac{d\bA_\star(a)}{da}\right|_{a=0} = \bmM_2^{-1} \lim_{x_1'\to 0} 
\int_{-\infty}^{0} \left\{\jump{0.15}{\bU}^\top(x_1'-x_1) \bR \frac{d^2\langle \bp \rangle(x_1)}{dx_1^2} 
+ \langle\bU\rangle^\top(x_1'-x_1) \bR \frac{d^2\jump{0.15}{\bp}(x_1)}{dx_1^2} \right\} dx_1.
$$

We assume here that all the derivatives exist, otherwise, as previously, the formula should be understood in the generalized sense.

\subsection{Computation of the perturbations of stress intensity factor according to the approach of Willis and Movchan (1995)} 
The Betti formula for the original crack problem and the perturbed crack problem reads
\beq
\lb{betti4}
\barr{l}
\jump{0.15}{\overline{\bU}}^{+\top} \bR \overline{\bsigma}_0^{+} - \overline{\bSigma}^{-\top} \bR \jump{0.15}{\overline{\bu}_0}^- = 
- \jump{0.15}{\overline{\bU}}^{+\top} \bR \langle \overline{\bp} \rangle - \langle\overline{\bU}\rangle^\top \bR \jump{0.15}{\overline{\bp}}, \quad \beta \in \Reals, \\[3mm]
\jump{0.15}{\overline{\bU}}^{+\top} \bR \overline{\bsigma}_\star^{\ddag} - \overline{\bSigma}^{-\top} \bR \jump{0.15}{\overline{\bu}_\star^{\ddag}} = 
- \jump{0.15}{\overline{\bU}}^{+\top} \bR \langle \overline{\bp} \rangle - \langle\overline{\bU}\rangle^\top \bR \jump{0.15}{\overline{\bp}}, \quad \beta \in \Reals,
\earr
\eeq
respectively. According to \eq{stress1} and \eq{displac}, the asymptotics of physical fields in the unperturbed problem are
\beq
\lb{unpert}
\barr{l}
\ds
\overline{\bsigma}_0^+(\beta) = \frac{\beta_+^{-1/2}}{4} \bmT_1(\beta) \bK_0 + \frac{\beta_+^{-1/2}}{4\beta} \bmT_2(\beta) \bA_0 + O(\beta^{-5/2}), \quad \beta \in \Complex^+, \\[3mm]
\ds
\jump{0.15}{\overline{\bu}_0}^-(\beta) = -\frac{bd_0^2}{4\beta} \beta_-^{-1/2} \bmV_1(\beta) \bK_0 - \frac{bd_0^2}{4\beta^2} \beta_-^{-1/2} \bmV_2(\beta) \bA_0 + O(\beta^{-7/2}), \quad 
\beta \in \Complex^-.
\earr
\eeq

For the physical fields in the perturbed problem, $\bsigma_\star^\ddag$ and $\bu_\star^\ddag$, we can write the respective transforms in the new 
coordinate system related to the new position of the crack tip
$$
\barr{l}
\ds
\overline{\bsigma}_\star^\ddag(\beta,a) =  
\int_a^{\infty} \bsigma_\star^\ddag(x_1) e^{i\beta x_1} dx_1 = 
e^{i\beta a} \int_0^{\infty} \bsigma_\star(y) e^{i\beta y} dy = 
e^{i\beta a} \overline{\bsigma}_\star^+(\beta,a), \\[3mm]
\ds
\jump{0.15}{\overline{\bu}_\star^\ddag}(\beta,a) =  
\int_{-\infty}^a \jump{0.15}{\bu_\star^\ddag}(x_1) e^{i\beta x_1} dx_1 = 
e^{i\beta a} \int_{-\infty}^0 \jump{0.15}{\bu_\star}(y) e^{i\beta y} dy = 
e^{i\beta a} \jump{0.15}{\overline{\bu}_\star}^-(\beta,a).
\earr
$$
Note that $\overline{\bsigma}_\star^+$ and $\jump{0.15}{\overline{\bu}_\star}^-$ are "$+$" and "$-$" functions, respectively, with the same asymptotics as in \eq{unpert}, 
subject to replacing $\bK_0$ with $\bK_\star(a)$ and $\bA_0$ with $\bA_\star(a)$.

Subtracting (\ref{betti4})$_2$ from (\ref{betti4})$_1$, we obtain
\beq
\lb{betti5}
\jump{0.15}{\overline{\bU}}^{+\top} \bR (\overline{\bsigma}_0^{+} - e^{i\beta a}\overline{\bsigma}_\star^{+}) 
- \overline{\bSigma}^{-\top} \bR (\jump{0.15}{\overline{\bu}_0}^- - e^{i\beta a}\jump{0.15}{\overline{\bu}_\star}^-) = \b0, 
\quad \beta \in \Reals.
\eeq 
Note that for any $a > 0$, the function $e^{i\beta a}\overline{\bsigma}_\star^{+}$ is a "$+$" function which decays exponentially as $\beta \to \infty$, 
$\beta \in \Complex^+$, whereas $e^{i\beta a}[\overline{\bu}_\star]^-$ is a "$-$" function which grows exponentially as $\beta \to \infty$, $\beta \in \Complex^-$.

Following Willis and Movchan (1995), we expand the exponential term as $a \to 0^+$, so that
$$
\jump{0.15}{\overline{\bU}}^{+\top} \bR (\overline{\bsigma}_0^{+} - (1 + ia\beta)\overline{\bsigma}_\star^{+}) 
- \overline{\bSigma}^{-\top} \bR (\jump{0.15}{\overline{\bu}_0}^- - (1 + ia\beta)\jump{0.15}{\overline{\bu}_\star}^-) = \b0, 
\quad \beta \in \Reals.
$$

The substitution of two-terms asymptotics for the unperturbed physical fields, $\overline{\bsigma}_0^{+}$, $\jump{0.15}{\overline{\bu}_0}^-$, and for the perturbed 
physical fields, $\overline{\bsigma}_\star^{+}$, $\jump{0.15}{\overline{\bu}_\star}^-$, leads to
$$
\barr{l}
\ds
\bmM_1 (\bK_0 - \bK_\star(a))\beta_+^{-1} - ia \bmM_1 \bK_\star(a) -ia \bmM_2 \bA_\star(a) \beta_+^{-1} \\[3mm]
\hspace{30mm}
- \{ \bmM_1(\bK_0 - \bK_\star(a))\beta_-^{-1} - ia \bmM_1 \bK_\star(a) -ia \bmM_2 \bA_\star(a) \beta_-^{-1} \} + O(\beta^{-2}) = \\[3mm]
\ds \hspace{30mm}
= \{\bmM_1 (\bK_0 - \bK_\star(a)) - ia \bmM_2 \bA_\star(a) \} (\beta_+^{-1} - \beta_-^{-1}) + O(\beta^{-2}) = \b0.
\earr
$$
We get 
\beq
\lb{SIF_perturb}
\Delta \bK = \bK_\star(a) - \bK_0 \sim -ia \bmM_1^{-1} \bmM_2 \bA_0 = 
a \left[\barr{ll} 1/2 + i\epsilon & 0 \\[3mm] 0 & 1/2 - i\epsilon \earr\right] \bA_0, \quad a \to 0^+,
\eeq
and finally
$$
\Delta K \sim a (1/2 + i\epsilon) A_0, \quad \Delta K^* \sim a (1/2 - i\epsilon) A_0^*, \quad a \to 0^+.
$$

Note that it is assumed here that $\bK_\star(a)$ and $\bA_\star(a)$ are continuous in $a = 0$. This fact will be proven in the next section.

\subsection{An alternative computation of the perturbations of stress intensity factors and high-order terms}
\lb{SIF_Abelian}
Here we provide a rigorous asymptotic procedure leading to the computation of perturbations in the stress intensity factors as well as high-order terms. The Abelian 
and Tauberian type theorems outlined in Appendix \ref{app3} will be used for this purpose.

We rewrite now the equation (\ref{betti5}) as follows 
\beq
\lb{w-h}
\jump{0.15}{\overline{\bU}}^{+\top} \bR (e^{-i\beta a}\overline{\bsigma}_0^{+} - \overline{\bsigma}_\star^{+}) 
- \overline{\bSigma}^{-\top} \bR (e^{-i\beta a}\jump{0.15}{\overline{\bu}_0}^- - \jump{0.15}{\overline{\bu}_\star}^-) = \b0, 
\quad \beta \in \Reals.
\eeq 
The function $e^{-i\beta a} \overline{\bSigma}^{-\top} \bR \jump{0.15}{\overline{\bu}_0}^{-}$ is a "$-$" function which decays exponentially as $\beta \to \infty$, 
$\beta \in \Complex^-$, whereas $e^{-i\beta a} \jump{0.15}{\overline{\bU}}^{+\top} \bR \overline{\bsigma}_0^+ $ is not a proper "$+$" function, because of the exponential 
growth as $\beta \to \infty$, $\beta \in \Complex^+$. However, for $\beta \in \Reals$, 
$e^{-i\beta a} \jump{0.15}{\overline{\bU}}^{+\top} \bR \overline{\bsigma}_0^+  \sim 1/\beta$, as $\beta \to \pm \infty$, so that it can be decomposed as 
$$
e^{-i\beta a} \jump{0.15}{\overline{\bU}}^{+\top} \bR \overline{\bsigma}_0^+ = \bXi^+ - \bXi^-,
$$
where
$$
\bXi^\pm(\beta) = \frac{1}{2\pi i} \int_{-\infty}^{\infty} 
\frac{e^{-ia t} \jump{0.15}{\overline{\bU}}^{+\top} \bR \overline{\bsigma}_0^+}{t - \beta} dt, 
\quad \beta \in \Complex^\pm.
$$
Now, from the solution of the Wiener-Hopf equation (\ref{w-h}), we obtain the following identities
\beq
\lb{id}
\barr{l}
\ds
\bXi^+ - \jump{0.15}{\overline{\bU}}^{+\top} \bR \overline{\bsigma}_\star^{+} = \b0, \quad \beta \in \Complex^+, \\[3mm]
\ds
-\bXi^- - \overline{\bSigma}^{-\top} \bR (e^{-i\beta a}\jump{0.15}{\overline{\bu}_0}^- - \jump{0.15}{\overline{\bu}_\star}^-) = \b0, \quad \beta \in \Complex^-.
\earr
\eeq
Note that
$$
\int_{-\infty}^{\infty} e^{-iat} \jump{0.15}{\overline{\bU}}^{+\top} \bR \overline{\bsigma}_0^+ dt < \infty,
$$
and
$$
e^{-i\beta a} \overline{\bSigma}^{-\top} \bR \jump{0.15}{\overline{\bu}_0}^- = o\left( \frac{1}{\beta^n} \right), 
\quad \beta \to \infty, \quad \beta \in \Complex^-,
$$
for any $n \in \Naturals$. Consequently, we obtain from the leading term asymptotic of both identities (\ref{id}) the same result, namely
\beq
\lb{k2}
\bmM_1 \bK_\star(a) = -\frac{1}{2\pi i} \int_{-\infty}^{\infty} e^{-iat} \jump{0.15}{\overline{\bU}}^{+\top} \bR \overline{\bsigma}_0^+ dt.
\eeq
The application of the Tauberian type Theorem \ref{theorem} (see Appendix \ref{app3}) leads to
$$
\lim_{a \to 0^+} \bmM_1 \bK_\star(a) = -\frac{1}{i} (- i \bmM_1 \bK_0),
$$
and, since the matrix $\bmM_1$ is not singular,
$$
\lim_{a \to 0^+} \bK_\star(a) = \bK_0.
$$
This proves that the complex stress intensity factor $K$ is continuous in $a = 0$.

In order to derive the first order variation of the complex stress intensity factor in the perturbation problem, we consider the derivative of (\ref{k2}) with 
respect to $a$ (see Theorem \ref{theorem}),
$$
\bmM_1 \bK_\star'(a) = \frac{1}{2\pi} \int_{-\infty}^{\infty} e^{-iat} \{t \jump{0.15}{\overline{\bU}}^{+\top} \bR \overline{\bsigma}_0^+ - \bmM_1 \bK_0\} dt,
$$
where
$$
t \jump{0.15}{\overline{\bU}}^{+\top} \bR \overline{\bsigma}_0^+ - \bmM_1 \bK_0 = t^{-1} \bmM_2 \bA_0 + t^{-2} \bmM_2 \bB_0 + O(t^{-3}), \quad t \to \infty, 
\quad t \in \Complex^+.
$$
Applying again the Theorem \ref{theorem} we can conclude that
$$
\lim_{a \to 0^+} \bmM_1 \bK_\star'(a) = -i \bmM_2 \bA_0, 
$$
and, finally, the first order variation of the complex stress intensity factor in the perturbation problem is given by 
$$
\bK_\star'(0) = -i \bmM_1^{-1} \bmM_2 \bA_0,
$$
which is consistent with the result \eq{SIF_perturb}.

This asymptotic procedure can be extended to compute the perturbations of high-order terms as follows. From the second-term asymptotics of (\ref{id}) we derive 
the coefficient $\bA_\star(a)$, which reads
\beq
\lb{A2}
\bmM_2 \bA_\star(a) = \lim_{\beta \to \infty} \frac{\beta}{2\pi i}
\int_{-\infty}^{\infty} \frac{t e^{-iat} \jump{0.15}{\overline{\bU}}^{+\top} \bR \overline{\bsigma}_0^+}{t - \beta} dt,
\quad \beta \in \Complex^+.
\eeq
Note that the integral on the right-hand side exists, but the behaviour of the integrand does not allow us to take the limit directly.

Let us consider the function
$$
\beffe(\beta,a) = \frac{\beta}{2\pi i}
\int_{-\infty}^{\infty} \frac{t e^{-iat} \jump{0.15}{\overline{\bU}}^{+\top} \bR \overline{\bsigma}_0^+}{t - \beta} dt.
$$

Integration by parts leads to
\beq
\lb{effe2}
\beffe(\beta,a) = \frac{\beta}{2\pi i} \left\{ 
\left[g(t,\beta) t \jump{0.15}{\overline{\bU}}^{+\top} \bR \overline{\bsigma}_0^+ \right]_{-\infty}^{\infty}
- \int_{-\infty}^{\infty} g(t,\beta) (t \jump{0.15}{\overline{\bU}}^{+\top} \bR \overline{\bsigma}_0^+ )' dt
\right\},
\eeq
where 
\beq
\lb{gi}
g(t,\beta) = \int_{-\infty}^{t} \frac{e^{-ia\xi}}{\xi - \beta} d\xi, \quad \beta \in \Complex^+.
\eeq
The integrand in (\ref{gi}) is analytic in $\Complex^-$ and decays exponentially as $\xi \to \infty$, 
$\xi \in \Complex^-$. As a result
$$
\lim_{t \to \infty} g(t,\beta) = 0, \quad \beta \in \Complex^+,
$$
and therefore the first term in brackets in \eq{effe2} vanishes and we can then write
$$
\barr{ll}
\ds \beffe(\beta,a) & \ds = -\frac{\beta}{2\pi i} \int_{-\infty}^{\infty} g(t,\beta) (t \jump{0.15}{\overline{\bU}}^{+\top} \bR \overline{\bsigma}_0^+ )' dt \\[3mm]
& \ds = -\frac{\beta}{2\pi i} \int_{-\infty}^{\infty} g(t,\beta) (t \jump{0.15}{\overline{\bU}}^{+\top} \bR \overline{\bsigma}_0^+ - \bmM_1 \bK_0)' dt \\[3mm]
& \ds = -\frac{\beta}{2\pi i} \left\{ 
\left[g(t,\beta) (t \jump{0.15}{\overline{\bU}}^{+\top} \bR \overline{\bsigma}_0^+ - \bmM_1 \bK_0)\right]_{-\infty}^{\infty} 
- \int_{-\infty}^{\infty} g'(t,\beta) (t \jump{0.15}{\overline{\bU}}^{+\top} \bR \overline{\bsigma}_0^+ - \bmM_1 \bK_0) dt 
\right\} \\[3mm]
& \ds = \frac{\beta}{2\pi i}  
\int_{-\infty}^{\infty} \frac{e^{-iat}}{t-\beta} (t \jump{0.15}{\overline{\bU}}^{+\top} \bR \overline{\bsigma}_0^+ - \bmM_1 \bK_0) dt.
\earr
$$
It is now possible to take the limit as $\beta \to \infty$ in (\ref{A2}), thus we deduce
$$
\barr{ll}
\ds \bmM_2 \bA_\star(a) & \ds = \lim_{\beta \to \infty} 
\frac{\beta}{2\pi i} \int_{-\infty}^{\infty} \frac{e^{-iat} (t \jump{0.15}{\overline{\bU}}^{+\top} \bR \overline{\bsigma}_0^+ - \bmM_1 \bK_0)}{t - \beta} dt \\[3mm]
& \ds = -\frac{1}{2\pi i} \int_{-\infty}^{\infty} e^{-iat} (t \jump{0.15}{\overline{\bU}}^{+\top} \bR \overline{\bsigma}_0^+ - \bmM_1 \bK_0) dt.
\earr
$$
We are now in the position to use the Theorem \ref{theorem}, so that
$$
\lim_{a \to 0^+} \bmM_2 \bA_\star(a) = -\frac{1}{i} (-i \bmM_2 \bA_0),
$$
and, since the matrix $\bmM_2$ is not singular,
$$
\lim_{a \to 0^+} \bA_\star(a) = \bA_0,
$$
which proves the continuity of the function $\bA_\star(a)$ in $a = 0$.

Moreover, the first order variation of $\bA_\star(a)$ can now be derived following the same asymptotic procedure presented above and we deduce
$$
\bA_\star'(0) = -i \bmM_2^{-1} \bmM_3 \bB_0,
$$
so that
$$
A_\star'(0) = (3/2 + i\epsilon) B_0, \quad A_\star^{*\prime}(0) = (3/2 - i\epsilon) B_0^*.
$$

\section{An illustrative example}
\lb{example}
In this section, we present an example regarding the computation of the complex stress intensity factor $K$ for an interfacial crack loaded by a simple asymmetric force 
system, as shown in Fig. \ref{three-point}. The loading is given by a point force $F$ acting upon the upper crack face at a distance $a$ behind the crack tip and two 
point forces $F/2$ acting upon the lower crack face at a distance $a - b$ and $a + b$ behind the crack tip. In terms of the Dirac delta function $\delta(\cdot)$, the 
loading is then given by
$$
\barr{l}
\ds p^+(x_1) = - F \delta(x_1 + a), \\[3mm]
\ds p^-(x_1) = -\frac{F}{2} \delta(x_1 + a + b) - \frac{F}{2} \delta(x_1 + a - b).
\earr
$$

%%%%%%%%%%%%%%%%%%%%%%%%%%%%%%%%%%%%%%%%%%%%%%%%%%%%%%%%%%%%%%%%%%%%%%
\begin{figure}[!htb]
\begin{center}
\vspace*{3mm}
\includegraphics[width=10cm]{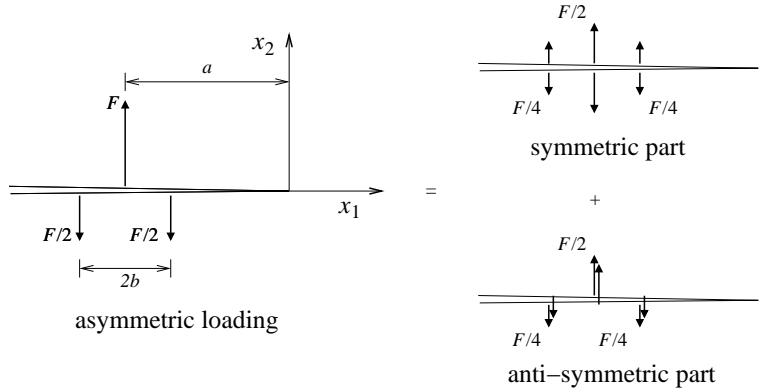}
\caption{\footnotesize "Three-point" loading of an interfacial crack.}
\label{three-point}
\end{center}
\end{figure}
%%%%%%%%%%%%%%%%%%%%%%%%%%%%%%%%%%%%%%%%%%%%%%%%%%%%%%%%%%%%%%%%%%%%%%

The loading is self-balanced, in terms of both principal force and moment vectors, and can be divided into symmetric and skew-symmetric parts (see Fig. 
\ref{three-point}), 
$$
\barr{l}
\ds \langle p \rangle(x_1) = -\frac{F}{2} \delta(x_1 + a) - \frac{F}{4} \delta(x_1 + a + b) - \frac{F}{4} \delta(x_1 + a - b), \\[3mm]
\ds \jump{0.15}{p}(x_1) = - F \delta(x_1 + a) + \frac{F}{2} \delta(x_1 + a + b) + \frac{F}{2} \delta(x_1 + a - b).
\earr
$$

For this loading system, the complex stress intensity factor $K$ has been evaluated by means of the integral formula \eq{SIF} and the symmetric and skew-symmetric 
weight functions, \eq{symm} and \eq{skew} respectively, obtaining $K = K^S + K^A$, where
$$
\barr{l}
\ds K^S = F \sqrt{\frac{2}{\pi}} \cosh(\pi\epsilon) a^{-1/2 - i\epsilon} \left\{ \frac{1}{2} + \frac{1}{4} (1 + b/a)^{-1/2 - i\epsilon} + 
\frac{1}{4} (1 - b/a)^{-1/2 - i\epsilon} \right\}, \\[3mm]
\ds K^A = \alpha F \sqrt{\frac{2}{\pi}} \cosh(\pi\epsilon) a^{-1/2 - i\epsilon} \left\{ \frac{1}{2} - \frac{1}{4} (1 + b/a)^{-1/2 - i\epsilon} -
\frac{1}{4} (1 - b/a)^{-1/2 - i\epsilon} \right\}.
\earr
$$

The complex stress intensity factor has been computed for the case $\nu_+ = 0.2$ and $\nu_- = 0.3$ and five values of the parameter $\eta$, namely 
$\eta = \{ -0.99, -0.5, 0, 0.5, 0.99 \}$. The values have been normalized and plotted in Fig. \ref{fig04} as functions of the ratio $b/a$ 
($a$ is fixed equal to 1). The symmetric (skew-symmetric) 
stress intensity factor is reported on the left (right) of the figure, the real (imaginary) part is reported on the top (bottom).
%%%%%%%%%%%%%%%%%%%%%%%%%%%%%%%%%%%%%%%%%%%%%%%%%%%%%%%%%%%%%%%%%%%%%%
\begin{figure}[!htb]
\begin{center}
\vspace*{3mm}
\includegraphics[width=11cm]{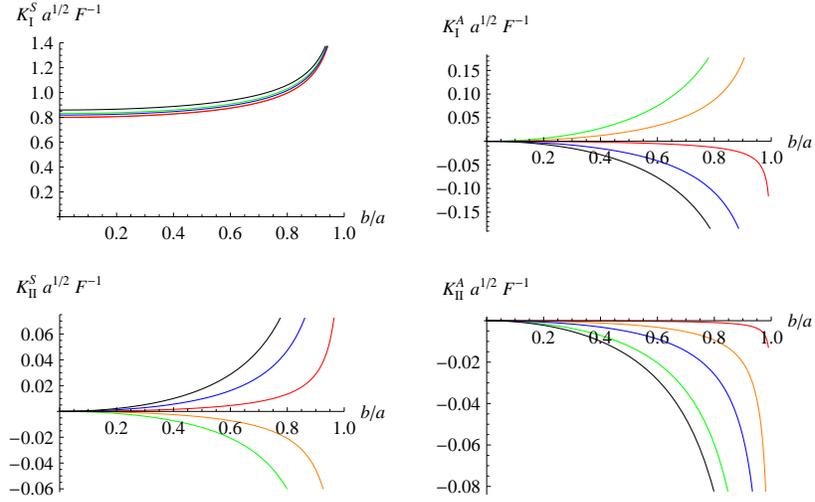}
\caption{\footnotesize Symmetric and anti-symmetric stress intensity factors as a function of $b/a$ ($a$ is fixed equal to 1) for the case 
$\nu_+ = 0.2$, $\nu_-=0.3$ and: $\eta = -0.99$ (green), $\eta = -0.5$ (orange), $\eta = 0$ (red), 
$\eta = 0.5$ (blue), $\eta = 0.99$ (black).}
\label{fig04}
\end{center}
\end{figure}
%%%%%%%%%%%%%%%%%%%%%%%%%%%%%%%%%%%%%%%%%%%%%%%%%%%%%%%%%%%%%%%%%%%%%%

Commenting on the results, first, we note from the figure that both $K_\modII^S$ and $K_\modII^A$ are not identically zero, even though the loading corresponds to a 
tensile (or opening) mode. This is typical of the interfacial crack problem, where there is no symmetry with respect to the plane containing the crack, so that Mode I and 
Mode II are coupled.

Second, results pertaining to the skew-symmetric stress intensity factor show that $K_\modI^A$ and $K_\modII^A$ are both equal to zero for $b/a = 0$, as expected, since 
in this case the loading is symmetric. As we increase $b/a$ and the skew-symmetric part of the loading becomes more and more relevant, the skew-symmetric stress intensity 
factor increases correspondingly. Both the symmetric and skew-symmetric stress intensity factors are singular as $b/a \to 1$, because a point force is approaching the 
crack tip.

The analytical method developed in the present paper allows us also to evaluate the coefficients in the higher order terms of asymptotics of physical 
fields. As an example, we computed the coefficient in the second order term, $A$, which is reported in Fig. \ref{fig05}, in 
analogy with the stress intensity factor in the figures above.
%%%%%%%%%%%%%%%%%%%%%%%%%%%%%%%%%%%%%%%%%%%%%%%%%%%%%%%%%%%%%%%%%%%%%%
\begin{figure}[!htb]
\begin{center}
\vspace*{3mm}
\includegraphics[width=11cm]{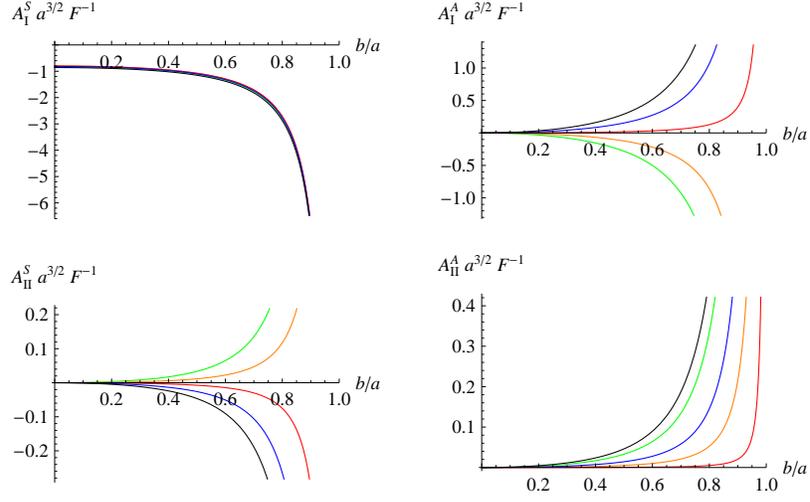}
\caption{\footnotesize Symmetric and anti-symmetric parts of $A$ as functions of $b/a$ ($a$ is fixed equal to 1) for the case $\nu_+ = 0.2$, 
$\nu_-=0.3$ and: $\eta = -0.99$ (green), $\eta = -0.5$ (orange), $\eta = 0$ (red), $\eta = 0.5$ (blue), $\eta = 0.99$ (black).}
\label{fig05}
\end{center}
\end{figure}
%%%%%%%%%%%%%%%%%%%%%%%%%%%%%%%%%%%%%%%%%%%%%%%%%%%%%%%%%%%%%%%%%%%%%%

In order to appreciate the magnitude of the skew-symmetric stress intensity factor with respect to the magnitude of the symmetric stress intensity factor, the ratio 
$K_\modI^A/K_\modI^S$ is plotted in Fig. \ref{fig06} (on the left) as a function of $b/a$. We notice from the figure that the magnitude of $K_\modI^A$ may 
easily reach 20\% of the magnitude of $K_\modI^S$, a value which is not negligible and has to be taken into account in the view of the application of a tensile fracture 
criterion. Note also that the coefficient $A$ is even more sensitive to the anti-symmetric loading, $A_\modI^A$ may 
easily reach the same magnitude as $A_\modI^S$, see Fig. \ref{fig06} on the right. The ratios $K_\modII^A/K_\modII^S$ and 
$A_\modII^A/A_\modII^S$ are not reported here since, in this case, they are both equal to the same constant, namely 
$K_\modII^A/K_\modII^S = A_\modII^A/A_\modII^S = \alpha$.
%%%%%%%%%%%%%%%%%%%%%%%%%%%%%%%%%%%%%%%%%%%%%%%%%%%%%%%%%%%%%%%%%%%%%%
\begin{figure}[!htb]
\begin{center}
\vspace*{3mm}
\includegraphics[width=11cm]{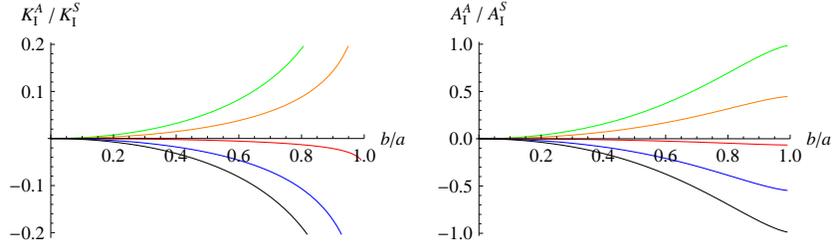}
\caption{\footnotesize Ratios $K_\modI^A/K_\modI^S$ and $A_\modI^A/A_\modI^S$ as functions of
$b/a$ ($a$ is fixed equal to 1) for the case 
$\nu_+ = 0.2$, $\nu_-=0.3$ and: $\eta = -0.99$ (green), $\eta = -0.5$ (orange), $\eta = 0$ (red), 
$\eta = 0.5$ (blue), $\eta = 0.99$ (black).}
\label{fig06}
\end{center}
\end{figure}
%%%%%%%%%%%%%%%%%%%%%%%%%%%%%%%%%%%%%%%%%%%%%%%%%%%%%%%%%%%%%%%%%%%%%%

\section{Mode III} 
\lb{antiplane}
The Wiener-Hopf equation relating the Fourier transforms of the singular displacement and the corresponding traction is given by 
$$
\jump{0.15}{\overline{U}_3}^+ = - \frac{b+e}{|\beta|} \overline{\Sigma}_{32}^-, \quad \beta \in \Reals,
$$
where $e = \nu_+ / \mu_+ + \nu_- / \mu_-$. This equation can be immediately factorized as follows
$$
\beta_+^{1/2} \jump{0.15}{\overline{U}_3}^+ = -\frac{b+e}{\beta_-^{1/2}} \overline{\Sigma}_{32}^-.
$$
Assuming that $\jump{0.15}{\overline{U}_3}^+ \sim \beta_+^{-1/2}$ and $\overline{\Sigma}_{32}^- \sim \beta_-^{1/2}$, as $\beta \to \infty$, $\beta \in \Complex^\pm$, 
we have
$$
\beta_+^{1/2} [\overline{U}_3]^+ = -\frac{b+e}{\beta_-^{1/2}} \overline{\Sigma}_{32}^- = O(1), \quad
\beta \to \pm\infty , \quad \beta \in \Reals,
$$
so that, from the Liouville theorem, it follows that both sides of the equation are equal to the same constant.

The weight functions for the Mode III problem are then given by
$$
\jump{0.15}{\overline{U}_3}^+(\beta) = \beta_+^{-1/2}, \quad \text{and} \quad 
\overline{\Sigma}_{32}^-(\beta) = -\frac{\beta_-^{1/2}}{b+e},
$$
with corresponding inverse transforms
\beq
\lb{jumpu3}
\jump{0.15}{U_3}(x_1) = \left\{
\barr{ll}
\ds \frac{1-i}{\sqrt{2\pi}} x_1^{-1/2}, & \ds x_1 > 0, \\[3mm]
0, & x_1 < 0,
\earr \right. \qquad
\Sigma_{32}(x_1) = \left\{
\barr{ll}
0, & \ds x_1 > 0, \\[3mm]
\ds \frac{1-i}{2\sqrt{2\pi}(b+e)} (-x_1)^{-3/2}, & x_1 < 0,
\earr \right.
\eeq

The Betti identity reduces to 
$$ 
\jump{0.15}{\overline{U}_3}^+ \overline{\sigma}_{32}^+ - \overline{\Sigma}_{32}^- \jump{0.15}{\overline{u}_3}^- = 
- \jump{0.15}{\overline{U}_3}^+ \langle \overline{p}_3 \rangle - \langle\overline{U}_3\rangle^+ \jump{0.15}{\overline{p}_3}.
$$

The asymptotics of physical fields are
$$
\barr{l}
\ds
\sigma_{32} = \frac{K_\modIII}{\sqrt{2\pi}} x_1^{-1/2} + \frac{A_\modIII}{\sqrt{2\pi}} x_1^{1/2} + O(x_1^{3/2}), \quad x_1 \to 0^+, \\[5mm]
\ds
\jump{0.15}{u_3} = \frac{2(b+e)K_\modIII}{\sqrt{2\pi}} (-x_1)^{1/2} - \frac{2(b+e)A_\modIII}{3\sqrt{2\pi}} (-x_1)^{3/2} + O[(-x_1)^{5/2}], \quad x_1 \to 0^-,
\earr
$$
with corresponding Fourier transforms
$$
\barr{l}
\ds
\overline{\sigma}_{32}^+ = \frac{(1+i)K_\modIII}{2} \beta_+^{-1/2} - \frac{(1-i)A_\modIII}{4} \beta_+^{-3/2} + O(\beta_+^{-5/2}), 
\quad \beta \to \infty, \quad \beta \in \Complex^+ \\[5mm]
\ds
\jump{0.15}{\overline{u}_3}^- = -\frac{(1+i)(b+e)K_\modIII}{2} \beta_-^{-3/2} + \frac{(1-i)(b+e)A_\modIII}{4} \beta_-^{-5/2} + O(\beta_-^{-7/2}), 
\quad \beta \to \infty, \quad \beta \in \Complex^-.
\earr
$$

By the same arguments used for Mode I and II, we obtain the formulae
$$
\barr{l}
\ds K_\modIII = -(1+i) \lim_{x_1' \to 0^+} \int_{-\infty}^{0} 
\{ \jump{0.15}{U_3}(x_1'-x_1) \langle p_3 \rangle(x_1) + \langle U_3 \rangle(x_1'-x_1) \jump{0.15}{p_3}(x_1) \} dx_1, \\[3mm]
\ds A_\modIII = -2(1+i) \lim_{x_1' \to 0^+} \int_{-\infty}^{0} 
\left\{ \jump{0.15}{U_3}(x_1'-x_1) \frac{d\langle p_3 \rangle(x_1)}{dx_1} + \langle U_3 \rangle(x_1'-x_1) \frac{d\jump{0.15}{p_3}(x_1)}{dx_1} \right\} dx_1.
\earr
$$

From the solution of the half-plane problem, we get the full representation of weight functions. For the lower half-plane:
$$
\barr{l}
\ds\overline{u}_3(\beta,x_2) = \frac{1}{\mu_-|\beta|} \overline{\Sigma}_{32}^- e^{|\beta|x_2}, \\[3mm]
\ds \overline{\sigma}_{31}(\beta,x_2) = -i \sign(\beta) \overline{\Sigma}_{32}^- e^{|\beta|x_2}, \\[3mm]
\ds \overline{\sigma}_{32}(\beta,x_2) = \overline{\Sigma}_{32}^- e^{|\beta|x_2}.
\earr
$$

For the upper half-plane the equations are the same subject to replacing $|\beta|$ with $-|\beta|$ and $\mu_-$ with $\mu_+$, so that we obtain
\beq
\lb{meanu3}
\jump{0.15}{\overline{U}_3}^+ = -\frac{\mu_+ + \mu_-}{\mu_ + \mu_-} \frac{1}{|\beta|}\overline{\Sigma}_{32}^-, \quad
\langle\overline{U}_3\rangle^+ = \frac{\mu_+ - \mu_-}{2\mu_+\mu_-} \frac{1}{|\beta|}\overline{\Sigma}_{32}^- = \frac{\eta}{2} \jump{0.15}{\overline{U}_3}^+,
\eeq 
where
$$
\eta = \frac{\mu_- - \mu_+}{\mu_- + \mu_+}.
$$

The formulae for the evaluation of $K_\modIII$ and $A_\modIII$ are 
\beq
\lb{k3}
\barr{l}
\ds K_\modIII = - \sqrt{\frac{2}{\pi}} \int_{-\infty}^{0} \left\{ \langle p_3 \rangle(x_1) + \frac{\eta}{2} \jump{0.15}{p_3}(x_1) \right\} (-x_1)^{-1/2} dx_1, \\[3mm]
\ds A_\modIII = -2 \sqrt{\frac{2}{\pi}} \int_{-\infty}^{0} 
\left\{ \frac{d\langle p_3 \rangle(x_1)}{dx_1} + \frac{\eta}{2} \frac{d\jump{0.15}{p_3}(x_1)}{dx_1} \right\} (-x_1)^{-1/2} dx_1.
\earr
\eeq

The full representation of weight functions after Fourier inversion is, for the lower half-plane, $x_2 < 0$,
$$
\barr{l}
\ds u_3(x_1,x_2) = -\frac{\mu_+}{2\sqrt{\pi}(\mu_- + \mu_+)}\left\{ (ix_1 - x_2)^{-1/2} - i(-ix_1 - x_2)^{-1/2} \right\}, \\[3mm]
\ds \sigma_{31}(x_1,x_2) = \frac{\mu_+\mu_-}{4\sqrt{\pi}(\mu_- + \mu_+)}\left\{ i(ix_1 - x_2)^{-3/2} - (-ix_1 - x_2)^{-3/2} \right\}, \\[3mm]
\ds \sigma_{32}(x_1,x_2) = -\frac{\mu_+\mu_-}{4\sqrt{\pi}(\mu_- + \mu_+)}\left\{ (ix_1 - x_2)^{-3/2} - i(-ix_1 - x_2)^{-3/2} \right\},
\earr
$$
and, for the upper half-plane, $x_2 > 0$,
$$
\barr{l}
\ds u_3(x_1,x_2) = \frac{\mu_-}{2\sqrt{\pi}(\mu_- + \mu_+)}\left\{ (ix_1 + x_2)^{-1/2} - i(-ix_1 + x_2)^{-1/2} \right\}, \\[3mm]
\ds \sigma_{31}(x_1,x_2) = -\frac{\mu_+\mu_-}{4\sqrt{\pi}(\mu_- + \mu_+)}\left\{ i(ix_1 + x_2)^{-3/2} - (-ix_1 + x_2)^{-3/2} \right\}, \\[3mm]
\ds \sigma_{32}(x_1,x_2) = -\frac{\mu_+\mu_-}{4\sqrt{\pi}(\mu_- + \mu_+)}\left\{ (ix_1 + x_2)^{-3/2} - i(-ix_1 + x_2)^{-3/2} \right\}.
\earr
$$

We compare now the result obtained for the Mode III with that of Mode I and II. In particular, we notice from \eq{k3} that the skew-symmetric part of $K_\modIII$ 
is proportional to the material parameter $\eta$, whereas the skew-symmetric part of the complex stress intensity factor $K$ is proportional to the material parameter 
$\alpha$, see \eq{SIF} and \eq{skew}. A comparison of these two parameters is then useful to understand whether the skew-symmetric loading is more relevant for Mode III 
($\alpha > \eta$) or Mode I and II ($\eta$ > $\alpha$). This comparison is shown in Fig. \ref{alpha}.
%%%%%%%%%%%%%%%%%%%%%%%%%%%%%%%%%%%%%%%%%%%%%%%%%%%%%%%%%%%%%%%%%%%%%%
\begin{figure}[!htb]
\begin{center}
\vspace*{3mm}
\includegraphics[width=6cm]{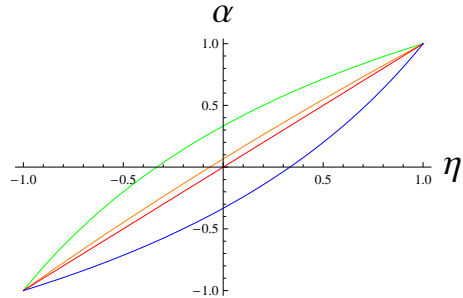}
\caption{\footnotesize The material parameter $\alpha$ as a function of the material parameter $\eta$ for: $\nu_+ = 0$, $\nu_-=0.5$ (green); 
$\nu_+ = 0.2$, $\nu_-=0.3$ (orange); $\nu_+ = \nu_-=0.3$ (red); 
$\nu_+ = 0.5$, $\nu_-=0$ (blue).}
\label{alpha}
\end{center}
\end{figure}
%%%%%%%%%%%%%%%%%%%%%%%%%%%%%%%%%%%%%%%%%%%%%%%%%%%%%%%%%%%%%%%%%%%%%%

\clearpage
%%%%%%%%%%%%%%%%%%%%%%%%%%%%%%%%%%%%%%%%%%%%%%%%%%%%%%%%%%%%%%%%%%%%%%%
%Appendix
\appendix
\renewcommand{\theequation}{\thesection.\arabic{equation}}

\section{APPENDIX}
\setcounter{equation}{0}

\subsection{Material parameters relevant to the interfacial crack problem}
\lb{app1}
The classical material parameters involved in the problem of the interfacial crack are the bimaterial constant,
$$
\epsilon = \frac{1}{2\pi} \log\frac{\mu_+ + (3-4\nu_+)\mu_-}{\mu_- + (3-4\nu_-)\mu_+},
$$
and the Dundurs parameters
$$
d_* = \frac{\mu_-(1 - 2\nu_+) - \mu_+(1 - 2\nu_-)}{2\mu_-(1 - \nu_+) + 2\mu_+(1 - \nu_-)}, \quad
\alpha = \frac{\mu_-(1 - \nu_+) - \mu_+(1 - \nu_-)}{\mu_-(1 - \nu_+) + \mu_+(1 - \nu_-)},
$$
(see for example Hutchinson, Mear and Rice, 1987).
Note that the bimaterial constant $\epsilon$ can be written in terms of the parameter $d_*$ as follow
$$
\epsilon = \frac{1}{\pi} \arctanh(d_*) = \frac{1}{2\pi} \log\frac{1+d_*}{1-d_*}.
$$
Commonly used are also the parameter
$$
b = \frac{1 - \nu_+}{\mu_+} + \frac{1 - \nu_-}{\mu_-},
$$
and, mainly for the 3D case, the parameter
$$
e = \frac{\nu_+}{\mu_+} + \frac{\nu_-}{\mu_-}.
$$
In the study of the asymmetrical case, new bimaterial parameters appear. In particular, the parameter
$$
\gamma = \frac{\mu_-(1 - 2\nu_+) + \mu_+(1 - 2\nu_-)}{2\mu_-(1 - \nu_+) + 2\mu_+(1 - \nu_-)},
$$
for the plane strain problem (Mode I and II), the parameter
$$
\eta = \frac{\mu_- - \mu_+}{\mu_- + \mu_+},
$$
for the antiplane shear problem (Mode III), and the parameter
$$
f = \frac{\nu_+}{\mu_+} - \frac{\nu_-}{\mu_-}
$$
for the 3D case.

All the parameters involved in the problem of the interfacial crack are listed in Table \ref{params}.
%%%%%%%%%%%%%%%%%%%%%%%%%%%%%%%%%%%%%%%%%%%%%%%%%%%%%%%%%%%%%%%%%%%%%% 
\begin{table}[!htcb]
\label{params}
\begin{center}
\caption{\footnotesize Material parameters relevant to the interfacial crack problem ($\kappa_\pm = 3 - 4\nu_\pm$).}
\vspace{3mm}
\begin{tabular}{||l|l||}
\hline\hline
 & \\[-2mm]
Symmetric weight function $\jump{0.15}{\bU}$ & Skew-symmetric weight function $\langle \bU \rangle$ \\[3mm]
\hline\hline
$\barr{ll}
  & \\
b & \ds = \frac{1 - \nu_+}{\mu_+} + \frac{1 - \nu_-}{\mu_-} \\[3mm]
  & \ds = \frac{1 + \kappa_+}{4\mu_+} + \frac{1 + \kappa_-}{4\mu_-} \\
  &
\earr$
&
$\barr{ll}
        & \\
b\alpha & \ds = \frac{1 - \nu_+}{\mu_+} - \frac{1 - \nu_-}{\mu_-} \\[3mm]
        & \ds = \frac{1 + \kappa_+}{4\mu_+} - \frac{1 + \kappa_-}{4\mu_-} \\
        & 
\earr$
\\
\hline
$\barr{ll}
  & \\
d & \ds = \frac{1 - 2\nu_+}{2\mu_+} - \frac{1 - 2\nu_-}{2\mu_-} \\[3mm]
  & \ds = \frac{\kappa_+ - 1}{4\mu_+} - \frac{\kappa_- - 1}{4\mu_-} \\
  &
\earr$
&
$\barr{ll}
        & \\
b\gamma & \ds = \frac{1 - 2\nu_+}{2\mu_+} + \frac{1 - 2\nu_-}{2\mu_-} \\[3mm]
        & \ds = \frac{\kappa_+ - 1}{4\mu_+} + \frac{\kappa_- - 1}{4\mu_-} \\
        &
\earr$
\\
\hline
$\barr{ll}
  & \\
e & \ds = \frac{\nu_+}{\mu_+} + \frac{\nu_-}{\mu_-} \\[3mm]
  & \ds = \frac{3 - \kappa_+}{4\mu_+} + \frac{3 - \kappa_-}{4\mu_-} \\
  &
\earr$
&
$\barr{ll}
  & \\
f & \ds = \frac{\nu_+}{\mu_+} - \frac{\nu_-}{\mu_-} \\[3mm] 
  & \ds = \frac{3 - \kappa_+}{4\mu_+} - \frac{3 - \kappa_-}{4\mu_-} \\
  &
\earr$
\\
\hline
$\barr{ll}
    & \\
d_* & \ds = \frac{d}{b} = \frac{\mu_-(1 - 2\nu_+) - \mu_+(1 - 2\nu_-)}{2\mu_-(1 - \nu_+) + 2\mu_+(1 - \nu_-)} \\[3mm]
    & \ds = \frac{\mu_-(\kappa_+ - 1) - \mu_+(\kappa_- - 1)}{\mu_-(1 + \kappa_+) + \mu_+(1 + \kappa_-)} \\
    &
\earr$
&
$\barr{ll}
         & \\
\gamma_* & \ds = \frac{\gamma}{\alpha} = \frac{\mu_-(1 - 2\nu_+) + \mu_+(1 - 2\nu_-)}{2\mu_-(1 - \nu_+) - 2\mu_+(1 - \nu_-)} \\[3mm]
         & \ds = \frac{\mu_-(\kappa_+ - 1) + \mu_+(\kappa_- - 1)}{\mu_-(1 + \kappa_+) - \mu_+(1 + \kappa_-)} \\
         &
\earr$
\\
\hline
&
$\barr{l}
 \\
\ds \eta = \frac{\mu_- - \mu_+}{\mu_- + \mu_+} \\
~
\earr$
\\
\hline
$\barr{l}
 \\
\ds \frac{1}{\mu_+} + \frac{1}{\mu_-} = b + e \\
~
\earr$
&
$\barr{l}
 \\
\ds \frac{1}{\mu_+} - \frac{1}{\mu_-} = b\alpha + f \\
~
\earr$
\\
\hline\hline
\end{tabular}
\end{center}
\end{table}
%%%%%%%%%%%%%%%%%%%%%%%%%%%%%%%%%%%%%%%%%%%%%%%%%%%%%%%%%%%%%%%%%%%%%%

We mention here that another bimaterial parameter, denoted by $\nu$, has been introduced by Lazarus and Leblond (1998), which is 
defined by
$$
1 - \nu = \frac{(1-\nu_+)/\mu_+ + (1-\nu_-)/\mu_-}{(1/\mu_+ + 1/\mu_-) \cosh^2(\pi\epsilon)},
$$
and can be expressed in our notations by
$$
1 - \nu = \frac{b d_0^4}{b + e}.
$$
As shown by Pindra {\em et al.}\ (2008), this bimaterial parameter has the meaning of ``equivalent Poisson ratio'', in the sense that it satisfies the inequalities 
$0 < \nu < 1/2$, and, when the two materials become the same, this quantity reduces to the Poisson ratio of the homogeneous body. 

There are several relations between material parameters that appear frequently in the solution of the problem. These relations are listed below
for convenience of the reader.

$$
d_0 = (1 - d_*^2)^{1/4}, \quad e_0 = e^{\pi\epsilon/2} = \left( \frac{1+d_*}{1-d_*} \right)^{1/4},
$$

$$
\cosh(\pi\epsilon) = \frac{1}{d_0^2}, \quad
\sinh(\pi\epsilon) = \frac{d_*}{d_0^2},
$$

$$
\cosh\left(\frac{\pi\epsilon}{2}\right) + \sinh\left(\frac{\pi\epsilon}{2}\right) = e_0, \quad
\cosh\left(\frac{\pi\epsilon}{2}\right) - \sinh\left(\frac{\pi\epsilon}{2}\right) = \frac{1}{e_0},
$$

$$
e_0^2 = \frac{1 + d_*}{d_0^2} = \frac{d_0^2}{1 - d_*}.
$$

\subsection{Abelian and Tauberian type theorems}
\lb{app3}
In this section, we will prove a proposition (Theorem \ref{theorem}), which is used in Section \ref{SIF_Abelian} for the purpose of evaluation of the perturbation 
of stress intensity factors and high-order analysis near the tip of the interfacial crack. This theorem appears to be a new result, not present in the literature.

We begin with presenting some known theorems. Although these results can be found in the literature (see for example Noble, 1988), they are presented here because of 
their relevance for the analysed 
problem. 

\begin{theorem} 
\lb{abel}
If, for some $a \in \Reals$, $|f(x)| < e^{ax}$ as $x \to +\infty$ and
\beq
\lb{est1}
f(x) \sim x^k, \quad x \to 0^+, \quad -1 < k < 0,
\eeq 
then 
$$
\psi^+(t) = \int_{0}^{\infty} f(x) e^{itx} dx
$$
is analytic in $\Complex_a^+ = \{t \in \Complex | \Im(t) > a\}$ and
\beq
\lb{est2}
\psi^+(t) \sim \Gamma(k+1) e^{i\frac{\pi}{2} (k+1)} t^{-k-1}, \quad t \to \infty, \quad 
\left|\arg(t) - \frac{\pi}{2}\right| \le A < \frac{\pi}{2}.
\eeq
\end{theorem}

\begin{theorem} 
\lb{tau}
Let $\psi^+(t)$ be a function analytic in $\Complex_0^+$, with the property (\ref{est2}) where it is assumed that $-1 < k < 0$. Then there exists a function $f$ 
with support $\Reals_+$,
$$
f(x) = \frac{1}{2\pi} \int_{-\infty}^{\infty} \psi^+(t) e^{-ixt} dt,
$$
satisfying the asymptotic estimate (\ref{est1}).
\end{theorem}

\begin{theorem} 
\lb{kappa}
Let $\psi$ be a locally integrable function on $\Reals$ satisfying, for some small $\epsilon > 0$,
\beq
\lb{conds}
\psi(t) = A_\pm |t|^{-k-1} + O(|t|^{-k-1-\epsilon}), \quad t \to \pm\infty, \quad -1 < k <0.
\eeq
Then there exists a function $f$,
$$
f(x) = \frac{1}{2\pi} \int_{-\infty}^{\infty} \psi(t) e^{-ixt} dt,
$$
satisfying 
\beq
\lb{res}
f(x) = B_\pm |x|^k + O(|x|^{k+\epsilon}), \quad x \to 0^\pm, 
\eeq
where
$$
B_\pm = \frac{\Gamma(-k)}{2\pi} \left\{ (A_+ + A_-)\sin\frac{\pi(k+1)}{2} \mp i(A_+ - A_-)\cos\frac{\pi(k+1)}{2} \right\}.
$$
\end{theorem}

\begin{proof} First we write
$$
\psi(t) = \psi_1(t) + \psi_2(t),
$$ 
where $\psi_1(t) = H(t) \psi(t)$ and $\psi_2(t) = [1 - H(t)] \psi(t)$, with $H(t)$ being the unit-step Heaviside function. Then, we consider the inverse transform of
$\psi_1(t)$ and, assuming a positive constant $a$, we deduce
$$
\barr{ll}
f_1(x) 
& \ds = \frac{1}{2\pi} \int_{0}^{\infty} \psi_1(t) e^{-ixt} dt \\[3mm]
& \ds = \frac{1}{2\pi} \left( \int_{0}^{a} + \int_{a}^{\infty} \right) \psi_1(t) e^{-ixt} dt \\[3mm]
& \ds = \frac{1}{2\pi} \int_{a}^{\infty} \psi_1(t) e^{-ixt} dt + O(1), \quad x \to 0,
\earr
$$
where the last equality is a consequence of the local integrability of the function $\psi_1$. From the asymptotics \eq{conds}$_1$, we have 
$\psi_1(t) \sim A_+ t^{-k-1}$, $t \to +\infty$ and thus
$$
f_1(x) = g(x) + O(x^{\min\{k+\epsilon,1\}}), \quad x \to 0,
$$
where
$$
g(x) = \frac{A_+}{2\pi} \int_{a}^{\infty} t^{-k-1} e^{-ixt} dt.
$$
Making the substitution $\xi = |x|t$, we obtain
\beq
\lb{gi2}
g(x) = \frac{A_+}{2\pi} |x|^k \int_{|x|a}^{\infty} \xi^{-k-1} e^{-i\sign(x)\xi} d\xi, 
\eeq
and hence
\beq
\lb{part1}
\barr{ll}
\ds \lim_{x \to 0^\pm} \frac{f_1(x)}{|x|^k} 
& \ds = \frac{A_+}{2\pi} \int_{0}^{\infty} \xi^{-k-1} e^{\mp i\xi} d\xi \\[5mm]
& \ds = \frac{A_+}{2\pi} \left\{ \int_{0}^{\infty} \xi^{-k-1} \cos\xi d\xi \mp i \int_{0}^{\infty} \xi^{-k-1} \sin\xi d\xi \right\} \\[5mm]
& \ds = \frac{A_+ \Gamma(-k)}{2\pi} \left\{ \sin\frac{\pi(k+1)}{2} \mp i \cos\frac{\pi(k+1)}{2} \right\}.
\earr
\eeq
Considering now the inverse transform of $\psi_2(t)$, we have
$$
f_2(x) = \frac{1}{2\pi} \int_{-\infty}^{0} \psi_2(t) e^{-ixt} dt = \frac{1}{2\pi} \int_{0}^{\infty} \psi_3(t) e^{ixt} dt,
$$
where $\psi_3(t) = \psi_2(-t) \sim A_- t^{-k-1}$, $t \to +\infty$. From the same argument used above for $f_1(x)$, we deduce
\beq
\lb{part2}
\lim_{x \to 0^\pm} \frac{f_2(x)}{|x|^k} = \frac{A_- \Gamma(-k)}{2\pi} \left\{ \sin\frac{\pi(k+1)}{2} \pm i \cos\frac{\pi(k+1)}{2} \right\}.
\eeq
It suffices now to add the partial results \eq{part1} and \eq{part2} to get the final result. 
\end{proof}

\begin{theorem} 
\lb{kappazero}
Let $\psi$ be a locally integrable function on $\Reals$ satisfying
\beq
\lb{conds3}
\psi(t) \sim A_\pm |t|^{-1}, \quad t \to \pm\infty, 
\eeq
and assume that $A_+ + A_- \ne 0$. Then there exists a function $f$,
$$
f(x) = \frac{1}{2\pi} \int_{-\infty}^{\infty} \psi(t) e^{-ixt} dt,
$$
satisfying 
\beq
\lb{res3}
f(x) \sim -\frac{1}{2\pi} (A_+ + A_-) \log |x|, \quad x \to 0.
\eeq
\end{theorem}

\begin{proof} The proof proceeds in the same way as for Theorem \ref{kappa}, except that the function $g(x)$ in \eq{gi2} is now given by
$$
g(x) = \frac{A_+}{2\pi} \int_{|x|a}^{\infty} \frac{1}{\xi} e^{-i\sign(x)\xi} d\xi 
= \frac{A_+}{2\pi} \left\{ -\Ci(|x|a) - i\sign(x) \left[\frac{\pi}{2} - \Si(|x|a)\right] \right\},
$$
where $\Ci$ and $\Si$ stand for the cosine and sine integral functions, respectively. From the properties $\Ci|x| \sim \log|x|$ and $\Si(x) \to 0$, as $x \to 0$, we 
deduce
$$
g(x) \sim -\frac{A+}{2\pi} \log|x|,
$$
and hence the assertion. 
\end{proof}

For the purpose of application of these results to the perturbation analysis of the physical fields around the interfacial crack, the need arises for an extension of 
the Theorem \ref{tau} to the case $k = 0$. This extension is provided by the following theorem. 

\begin{theorem} 
\lb{theorem}
Let $f(x)$ be the function
\beq
\lb{effe}
f(x) = \frac{1}{2\pi} \int_{-\infty}^{\infty} \psi^+(t) e^{-ixt} dt.
\eeq
If $\psi^+(t)$ is analytic in $\Complex^+$ and
\beq
\lb{assumpt}
\psi^+(t) = a_1 t^{-1} + a_2 t^{-2} + O\left(t^{-3}\right), \quad t \to \infty, 
\eeq
in the closed half-plane $\overline{\Complex}^+ = \Complex^+ \cup \Reals$, then $f(x) = 0,\ \forall x < 0$ and 
\beq
\lb{res1}
\lim_{x \to 0^+} f(x) = - i a_1, 
\eeq
\beq
\lb{res2}
f'(x) = -\frac{i}{2\pi} \int_{-\infty}^{\infty} \{t \psi^+(t) - a_1\}e^{-ixt} dt.
\eeq
\end{theorem}

\begin{proof} The fact that $f(x) = 0,\ \forall x < 0$ is a direct consequence of the fact that $\psi^+(t)$ is a "$+$" function. Assume now that $x > 0$. 

From the assumptions on the behaviour of the function $\psi^+(t)$, it follows that $\psi^+(t) = a_1/t + R(t)$, where $t R(t) \to 0$, as $t \to \infty$, $t \in \overline{\Complex}^+$ 
(including $t \to \pm\infty$, $t \in \Reals$).

We write
\beq
\lb{int}
f(x) = \frac{1}{2\pi} \lim_{a \to +\infty} \left\{ \int_{a}^{\infty} [\psi^+(-t) e^{ixt} + \psi^+(t) e^{-ixt}] dt + \int_{-a}^{a} \psi^+(t) e^{-ixt} dt \right\}.
\eeq
The first integral is
$$
f_1(x,a) = \int_{a}^{\infty} [\psi^+(-t) e^{ixt} + \psi^+(t) e^{-ixt}] dt = f_{11}(x,a) + f_{12}(x,a),
$$
where
$$
f_{11}(x,a) = \int_{a}^{\infty} \left[-\frac{a_1}{t} e^{ixt} + \frac{a_1}{t} e^{-ixt}\right] dt 
= -2i a_1 \int_{xa}^{\infty} \frac{\sin(\xi)}{\xi} d\xi,
$$
and
$$
f_{12}(x,a) = \int_{a}^{\infty} \left[R(-t) e^{ixt} + R(t) e^{-ixt}\right] dt 
= \int_{xa}^{\infty} \left[ \frac{1}{x}R\left(-\frac{\xi}{x}\right)e^{i\xi} + \frac{1}{x}R\left(\frac{\xi}{x}\right)e^{-i\xi} \right] d\xi,
$$
so that, taking $a = x^{-1/2}$, we have $f_{11}(x,x^{-1/2}) \to -i\pi a_1$ and $f_{12}(x,x^{-1/2}) \to 0$, as $x \to 0^+$.

Let us denote the second integral in \eq{int} by $f_2(x,a)$. Then, using analyticity of $\psi^+(t)$ in $\Complex^+$, we deduce
$$
f_2(x,a) = - \int_{\Gamma_a} \psi^+(t) e^{-ixt} dt,
$$
where $\Gamma_a = \{t \in \Complex | t = a e^{i\theta}, 0 < \theta < \pi\}$. We write 
$$
f_2(x,a) = f_{21}(x,a) + f_{22}(x,a),
$$
where
$$
f_{21}(x,a) = - \int_{\Gamma_a} \frac{a_1}{t} e^{-ixt} dt \quad \text{and} \quad 
f_{22}(x,a) = - \int_{\Gamma_a} R(t) e^{-ixt} dt.
$$
Taking again the same $a = x^{-1/2}$, we obtain
$$
f_{21}(x,x^{-1/2}) = - \int_{\Gamma_{x^{-1/2}}} \frac{a_1}{t} e^{-ixt} dt \sim - a_1 \int_{\Gamma_{x^{-1/2}}} \frac{1}{t} dt = -i \pi a_1,
$$
as $x \to 0^+$.

Finally, 
$$
f_{22}(x,a) = - \int_{\Gamma_a} R(t) e^{-ixt} dt = - \int_{\Gamma_{xa}} \frac{1}{x}R\left(\frac{\xi}{x}\right) e^{-i\xi} d\xi,
$$
so that $f_{22}(x,x^{-1/2}) \to 0$, as $x \to 0^+$, and we can conclude that 
$$
f(x) \to \frac{1}{2\pi} (-i\pi a_1 - i\pi a_1) = -i a_1,
$$
as $x \to 0^+$, which proves \eq{res1}.

Taking the derivative of \eq{effe} with respect to $x$ we obtain 
$$
f'(x) = -\frac{i}{2\pi} \int_{-\infty}^{\infty} t\psi^+(t) e^{-ixt} dt.
$$
Note that this integral does not exist for real $x$, so everywhere later we assume $\Im(x) > 0$. Integrating by parts 
$$
\barr{ll}
f'(x) 
& \ds = -\frac{i}{2\pi} \left\{ \left[ t\psi^+(t) \int_{-\infty}^{t} e^{-ixp} dp \right]_{-\infty}^{\infty} 
- \int_{-\infty}^{\infty} [t\psi^+(t)]' \left(\int_{-\infty}^{t} e^{-ixp} dp\right) dt \right\} \\[5mm]
& \ds = \frac{i}{2\pi} \int_{-\infty}^{\infty} [t\psi^+(t) - a_1]' \left(\int_{-\infty}^{t} e^{-ixp} dp\right) dt,
\earr
$$
because $\int_{-\infty}^{t} e^{-ixp} dp \to 0$ as $t \to \pm\infty$. Integrating again by parts 
$$
\barr{ll}
f'(x) 
& \ds = \frac{i}{2\pi} \left\{ \left[ [t\psi^+(t) - a_1] \int_{-\infty}^{t} e^{-ixp} dp \right]_{-\infty}^{\infty} 
- \int_{-\infty}^{\infty} [t\psi^+(t) - a_1] e^{-ixt} dt \right\} \\[5mm]
& \ds = -\frac{i}{2\pi} \int_{-\infty}^{\infty} [t\psi^+(t) - a_1] e^{-ixt} dt.
\earr
$$
Note that the last formula makes sense also for real $x$. 
\end{proof}

An immediate consequence of the Theorem \ref{theorem} is the following

\begin{corollarynonum}
It follows from the Theorem \ref{theorem} that
$$
\lim_{x \to 0^+} f'(x) = - a_2.
$$
\end{corollarynonum}

\begin{remarknonum}
Note that the difference between the Theorems \ref{theorem} and \ref{tau} is in the assumption about the asymptotic behaviour of the transform $\psi^+$ 
as $t \to \infty$.
The assumption \eq{est2} is weaker, since it is equivalent to assume the asymptotic property as $\Im(t) \to +\infty$, whereas the assumption \eq{assumpt} is stronger, 
since the asymptotic property must hold true also as $t \to \infty$ along the real axis ($\Im(t) = 0$ and $\Re(t) \to \pm\infty$). The stronger assumption is needed in 
order to extend the theorem to the case $k = 0$.

Note also that analyticity of the function $\psi^+$ in the Theorem \ref{theorem} is essential. To make this clear let us consider the following examples.

In the case $k = 0$ and $A_+ + A_- = 0$, neither Theorem \ref{kappa} nor Theorem \ref{kappazero} are applicable and the only conclusion on the asymptotic 
behaviour of the inverse transform $f(x)$ is that $f(x)$ is asymptotically equal to a constant as $x \to 0$, but the value of this constant depends, in general, on 
the particular function $\psi(t)$, not only on its asymptotic behaviour as $t \to \pm\infty$. For example, if $\psi(t) = t/(t^2 + 1)$ then $f(x) \to \mp i/2$ as 
$x \to 0^\pm$, whereas if $\psi(t) = (t+1)/(t^2 + 1)$ then $f(x) \to 1/2 \mp i/2$ as $x \to 0^\pm$. From these examples, it is clear that smoothness of the function 
$\psi(t)$ is not helpful.

Analyticity of the function $\psi(t) \equiv \psi^+(t)$ in $\Complex^+$, in addition to conditions (\ref{conds}) guarantees the result (\ref{res1}) for $k = 0$ and 
$A_+ + A_- = 0$.
\end{remarknonum}

\end{document}

%% file: definitions.tex
\newcommand{\beq}{\begin{equation}}
\newcommand{\eeq}{\end{equation}}
\newcommand{\lb}{\label}
\newcommand{\beqar}{\begin{eqnarray}}
\newcommand{\eeqar}{\end{eqnarray}}
\newcommand{\bit}{\begin{itemize}}
\newcommand{\eit}{\end{itemize}}
\newcommand{\benum}{\begin{enumerate}}
\newcommand{\eenum}{\end{enumerate}}
\newcommand{\barr}{\begin{array}}
\newcommand{\earr}{\end{array}}
\def\ds{\displaystyle}
\newcommand\eq[1]{(\ref{#1})}
\newtheorem{theorem}{Theorem}[section]

\newtheorem*{remarknonum}{Remark}

\newtheorem*{corollarynonum}{Corollary}

\newcommand{\jump}[2]{
[\mbox{\hspace{-#1em}}[#2]\mbox{\hspace{-#1em}}]}
\newcommand{\deriv}[2]{\frac{\partial #1}{\partial #2}}
\newcommand{\nderiv}[3]{\frac{\partial^{\,#3} #1}{\partial #2^{\,#3}}}

\newcommand{\modI}{\text{I}}
\newcommand{\modII}{\text{II}}
\newcommand{\modIII}{\text{III}}

\def\XXint#1#2#3{{\setbox0=\hbox{$#1{#2#3}{\int}$}
   \vcenter{\hbox{$#2#3$}}\kern-.5\wd0}}

\def\b0{\mbox{\boldmath $0$}}

\def\be{\mbox{\boldmath $e$}}
\def\beffe{\mbox{\boldmath $f$}}

\def\bp{\mbox{\boldmath $p$}}

\def\bu{\mbox{\boldmath $u$}}

\def\bw{\mbox{\boldmath $w$}}

\def\bA{\mbox{\boldmath $A$}}
\def\bB{\mbox{\boldmath $B$}}

\def\bD{\mbox{\boldmath $D$}}

\def\bF{\mbox{\boldmath $F$}}
\def\bG{\mbox{\boldmath $G$}}

\def\bK{\mbox{\boldmath $K$}}

\def\bQ{\mbox{\boldmath $Q$}}
\def\bR{\mbox{\boldmath $R$}}
\def\bS{\mbox{\boldmath $S$}}
\def\bT{\mbox{\boldmath $T$}}
\def\bU{\mbox{\boldmath $U$}}
\def\bV{\mbox{\boldmath $V$}}

\def\bX{\mbox{\boldmath $X$}}

\newcommand{\bsigma}{\mbox{\boldmath $\sigma$}}
\newcommand{\bSigma}{\mbox{\boldmath $\Sigma$}}

\newcommand{\bPsi}{\mbox{\boldmath$\Psi$}}

\newcommand{\bPhi}{\mbox{\boldmath $\Phi$}}
\newcommand{\bXi}{\mbox{\boldmath $\Xi$}}

\def\f0{\ensuremath{\mathbb{O}}}

\newcommand{\Ga}{\alpha}

\newcommand{\mF}{\ensuremath{\mathcal{F}}}

\newcommand{\mS}{\ensuremath{\mathcal{S}}}

\newcommand{\mU}{\ensuremath{\mathcal{U}}}

\newcommand{\bmB}{\mbox{\boldmath $\mathcal{B}$}}

\newcommand{\bmM}{\mbox{\boldmath $\mathcal{M}$}}

\newcommand{\bmS}{\mbox{\boldmath $\mathcal{S}$}}
\newcommand{\bmT}{\mbox{\boldmath $\mathcal{T}$}}
\newcommand{\bmU}{\mbox{\boldmath $\mathcal{U}$}}
\newcommand{\bmV}{\mbox{\boldmath $\mathcal{V}$}}

\def\Im{\mathop{\mathrm{Im}}}
\def\Re{\mathop{\mathrm{Re}}}

\newcommand{\Ci}{\mathop{\mathrm{Ci}}}
\newcommand{\Si}{\mathop{\mathrm{Si}}}
\newcommand{\sign}{\mathop{\mathrm{sign}}}
\newcommand{\Res}{\mathop{\mathrm{Res}}}
\newcommand{\arctanh}{\mathop{\mathrm{arctanh}}}

\newcommand{\Reals}{\ensuremath{\mathbb{R}}}
\newcommand{\Complex}{\ensuremath{\mathbb{C}}}
\newcommand{\Naturals}{\ensuremath{\mathbb{N}}}

\def\IJSS{{\it Int.\ J.\ Solids Struct.}\ }

\def\JAM{{\it ASME J.\ Appl.\ Mech.}\ }

\def\JE{{\it J.\ Elasticity}\ }

\def\JMPS{{\it J.\ Mech.\ Phys.\ Solids}\ }

\def\ZAMM{{\it Z.\ Angew.\ Math.\ Mech.}\ }